\documentclass[11pt,amssymb,amsfont,a4paper]{article}
\usepackage{latexsym,amssymb,amsmath,amsthm,color}
\usepackage{geometry}
\usepackage{url}
\usepackage{graphicx}
\usepackage[utf8]{inputenc}
\usepackage[numbers]{natbib}
\usepackage{authblk}

\usepackage{tikz}
\usetikzlibrary{angles,calc,intersections,quotes,arrows.meta}
\usetikzlibrary{calc, arrows,backgrounds,positioning,fit}
\usetikzlibrary{decorations.markings}
\usetikzlibrary{shapes.geometric}
\usetikzlibrary{shapes.misc, fit}
\usetikzlibrary{decorations.pathreplacing}
\usepackage{pgfplots}
\tikzset{middlearrow/.style={
        decoration={markings,
            mark= at position 0.6 with {\arrow{#1}} ,
        },
        postaction={decorate}
    }
}

\usetikzlibrary{shapes}
\usepackage{arydshln}

\theoremstyle{definition}
\newtheorem{definition}{Definition}
\newtheorem{example}[definition]{Example}
\newtheorem{construction}{Construction}

\theoremstyle{plain}
\newtheorem{theorem}{Theorem}
\newtheorem{proposition}[definition]{Proposition}
\newtheorem{lemma}[definition]{Lemma}
\newtheorem{remark}[definition]{Remark}
\newtheorem{corollary}[definition]{Corollary}

\title{Multilayer crisscross error and erasure correction}
\author{Umberto Mart{\'i}nez-Pe\~{n}as \thanks{umberto.martinez@uva.es}}
\affil{IMUVa-Mathematics Research Institute,\\University of Valladolid, Spain}

\date{}

\begin{document}

\maketitle

\begin{abstract}
In this work, multilayer crisscross error and erasures are considered, which affect entire rows and columns in the matrices of a list of matrices. To measure such errors and erasures, the multi-cover metric is introduced. Several bounds are derived, including a Singleton bound, and maximum multi-cover distance (MMCD) codes are defined as those attaining it. Duality, puncturing and shortening of linear MMCD codes are studied. It is shown that the dual of a linear MMCD code is not necessarily MMCD, and those satisfying this duality condition are defined as dually MMCD codes. Finally, some constructions of codes in the multi-cover metric are given, including dually MMCD codes, together with efficient decoding algorithms for them.
%

\end{abstract}

\section{Introduction} \label{sec intro}

The cover metric was introduced independently in \cite{gabidulin-lattice, gabidulin-crisscross, roth} to measure the number of crisscross errors in memory chip arrays \cite{roth}. Later they proved to be useful in narrow-band power line communication (NB-PLC) smart grids \cite{kabore, yazbek}. Crisscross errors affect entire rows and columns of a matrix that encodes the stored information. Two types of codes attaining the Singleton bound for the cover metric were introduced in \cite{gabidulin-crisscross, roth}, one of them based on the rank metric and patented in \cite{crisscross-patent}. Since then, codes in the cover metric with other correcting features (e.g., probabilistic, list-decoding, etc.) were subsequently studied \cite{bitar-crisscross, kadhe-crisscross, liu-crisscross, roth-probabilistic, wachter-crisscross, welter-crisscross}.

In this work, we extend the cover metric to tuples of $ \ell $ matrices or arrays, where the length $ \ell $ is a fixed positive integer. We call the new metric the multi-cover metric. This metric is suitable to correct simultaneously a number of crisscross errors and erasures distributed over the $ \ell $ matrices, in any way and without knowledge at the decoding end. We call such errors and erasures multilayer crisscross errors and erasures. These patterns of errors and erasures occur naturally when one considers $ \ell $ independent realizations of channels or systems where crisscross errors and erasures occur. This is the case, for instance, when one considers correction accross multiple memory chip arrays simultaneously, or when one considers multiple NB-PLC smart grids. In such situations, considering a Cartesian product of codes attaining the Singleton bound for the cover metric or a concatenation with convolutional codes do not yield codes attaining the Singleton bound for the multi-cover metric. Hence, if one wants optimal codes (meaning those whose minimum multi-cover distance attains the Singleton bound), one needs to consider other methods (see the bounds and constructions in Sections \ref{sec bounds} and \ref{sec constructions}, respectively). See Fig. \ref{fig error pattern} below for a typical multilayer crisscross error pattern.

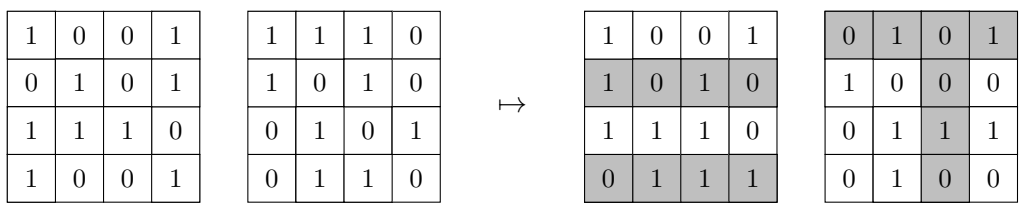
\begin{figure} [!h]
\hspace*{-1.8em}
\begin{center}
\begin{tabular}{c@{\extracolsep{1cm}}c}
\begin{tikzpicture}[
square/.style = {draw, rectangle, 
                 minimum size=\m, outer sep=0, inner sep=0, font=\small,
                 },
                        ]
\def\m{18pt}
\def\w{4}
\def\h{4}
    \pgfmathsetmacro\uw{int(\w/2)}
    \pgfmathsetmacro\uh{int(\h/2)}


\def\i{10}
  \foreach \x in {1,...,4}
    \foreach \y in {1,...,\h}
       \node [square, fill=white]  (\x,\y) at (\x*\m + \i*\m,-\y*\m) {$ 0 $};
       
	\node [square, fill=white] () at (1*\m + \i*\m,-1*\m) {$ 1 $};	
	\node [square, fill=white] () at (1*\m + \i*\m,-3*\m) {$ 1 $};	
	\node [square, fill=white] () at (1*\m + \i*\m,-4*\m) {$ 1 $};	
	\node [square, fill=white] () at (2*\m + \i*\m,-2*\m) {$ 1 $};
	\node [square, fill=white] () at (2*\m + \i*\m,-3*\m) {$ 1 $};
	\node [square, fill=white] () at (3*\m + \i*\m,-3*\m) {$ 1 $};
	\node [square, fill=white] () at (4*\m + \i*\m,-4*\m) {$ 1 $};
	\node [square, fill=white] () at (4*\m + \i*\m,-1*\m) {$ 1 $};
	\node [square, fill=white] () at (4*\m + \i*\m,-2*\m) {$ 1 $};
	
	
  \foreach \x in {6,...,9}
    \foreach \y in {1,...,\h}
       \node [square, fill=white]  (\x,\y) at (\x*\m + \i*\m,-\y*\m) {$ 0 $};

	\node [square, fill=white] () at (5*\m + 1*\m + \i*\m,-1*\m) {$ 1 $};	
	\node [square, fill=white] () at (5*\m + 1*\m + \i*\m,-2*\m) {$ 1 $};	
	\node [square, fill=white] () at (5*\m + 2*\m + \i*\m,-1*\m) {$ 1 $};	
	\node [square, fill=white] () at (5*\m + 2*\m + \i*\m,-3*\m) {$ 1 $};
	\node [square, fill=white] () at (5*\m + 2*\m + \i*\m,-4*\m) {$ 1 $};
	\node [square, fill=white] () at (5*\m + 3*\m + \i*\m,-1*\m) {$ 1 $};
	\node [square, fill=white] () at (5*\m + 3*\m + \i*\m,-2*\m) {$ 1 $};
	\node [square, fill=white] () at (5*\m + 3*\m + \i*\m,-4*\m) {$ 1 $};
	\node [square, fill=white] () at (5*\m + 4*\m + \i*\m,-3*\m) {$ 1 $};
	
	
\node [] () at (11*\m + \i*\m, -\h * \m / 2 - \m / 2) {$ \mapsto $};
	
	
  \foreach \x in {13,...,16}
    \foreach \y in {1,...,\h}
       \node [square, fill=white]  (\x,\y) at (\x*\m + \i*\m,-\y*\m) {$ 0 $};
       
  \foreach \x in {13,...,16}
    \foreach \y in {2, 4}
       \node [square, fill=gray!50]  (\x,\y) at (\x*\m + \i*\m,-\y*\m) {$ 1 $};
       
	\node [square, fill=white] () at (12*\m + 1*\m + \i*\m,-1*\m) {$ 1 $};	
	\node [square, fill=white] () at (12*\m + 1*\m + \i*\m,-3*\m) {$ 1 $};	
	\node [square, fill=gray!50] () at (12*\m + 1*\m + \i*\m,-4*\m) {$ 0 $};	
	\node [square, fill=gray!50] () at (12*\m + 2*\m + \i*\m,-2*\m) {$ 0 $};
	\node [square, fill=white] () at (12*\m + 2*\m + \i*\m,-3*\m) {$ 1 $};
	\node [square, fill=white] () at (12*\m + 3*\m + \i*\m,-3*\m) {$ 1 $};
	\node [square, fill=gray!50] () at (12*\m + 4*\m + \i*\m,-4*\m) {$ 1 $};
	\node [square, fill=white] () at (12*\m + 4*\m + \i*\m,-1*\m) {$ 1 $};
	\node [square, fill=gray!50] () at (12*\m + 4*\m + \i*\m,-2*\m) {$ 0 $};
	
	
  \foreach \x in {18,...,21}
    \foreach \y in {1,...,\h}
       \node [square, fill=white]  (\x,\y) at (\x*\m + \i*\m,-\y*\m) {$ 0 $};
       
  \foreach \x in {18,...,21}
    \foreach \y in {1}
       \node [square, fill=gray!50]  (\x,\y) at (\x*\m + \i*\m,-\y*\m) {$ 1 $};
       
  \foreach \x in {20}
    \foreach \y in {1,...,\h}
       \node [square, fill=gray!50]  (\x,\y) at (\x*\m + \i*\m,-\y*\m) {$ 1 $};
       
	\node [square, fill=gray!50] () at (12*\m + 5*\m + 1*\m + \i*\m,-1*\m) {$0 $};	
	\node [square, fill=white] () at (12*\m + 5*\m + 1*\m + \i*\m,-2*\m) {$ 1 $};	
	\node [square, fill=gray!50] () at (12*\m + 5*\m + 2*\m + \i*\m,-1*\m) {$ 1 $};	
	\node [square, fill=white] () at (12*\m + 5*\m + 2*\m + \i*\m,-3*\m) {$ 1 $};
	\node [square, fill=white] () at (12*\m + 5*\m + 2*\m + \i*\m,-4*\m) {$ 1 $};
	\node [square, fill=gray!50] () at (12*\m + 5*\m + 3*\m + \i*\m,-1*\m) {$ 0 $};
	\node [square, fill=gray!50] () at (12*\m + 5*\m + 3*\m + \i*\m,-2*\m) {$ 0 $};
	\node [square, fill=gray!50] () at (12*\m + 5*\m + 3*\m + \i*\m,-4*\m) {$ 0 $};
	\node [square, fill=white] () at (12*\m + 5*\m + 4*\m + \i*\m,-3*\m) {$ 1 $};
        
\end{tikzpicture}

\end{tabular}
\end{center}

\caption{Illustration of a pattern of crisscross errors throughout two $ 4 \times 4 $ binary matrices (which constitute one single codeword). The error pattern affects almost entirely (but not necessarily entirely) two rows in the first matrix and one column and one row in the second matrix. Thus in this example we are considering $ \ell = 2 $ realizations of a channel with crisscross errors, where a total $ 4 $ multilayer crisscross errors occured ($ 2 $ in the first realization and $ 2 $ in the second one). }
\label{fig error pattern}

\end{figure}

The organization and main contributions of this manuscript are as follows. In Section \ref{sec definitions}, we provide the main definitions, notations and some preliminary results. In Section \ref{sec bounds}, we provide several bounds on the parameters of codes in the multi-cover metric. These include a Singleton-like bound. Maximum multi-cover distance (MMCD) codes are then defined as those attaining it. In Section \ref{sec duality, puncturing and shortening}, we explore the duality, puncturing and shortening of linear MMCD codes. It is shown that the dual of a linear MMCD code is not MMCD in general, in contrast with other classical metrics. Dually MMCD codes are defined as those linear MMCD codes whose duals are also MMCD. We characterize dually MMCD codes in terms of information multi-covers, which generalize the notion of information sets. In Section \ref{sec constructions}, we provide several constructions of codes in the multi-cover metric, including dually MMCD codes, and provide decoders for them. Section \ref{sec conclusion} concludes the paper and provides several open problems.

\section{Definitions and basic properties} \label{sec definitions}

Denote $ \mathbb{N} = \{ 0,1,2, \ldots \} $ and fix $ q $ a prime power. We denote by $ \mathbb{F}_q $ the finite field with $ q $ elements. For positive integers $ m $ and $ n $, we denote by $ \mathbb{F}_q^{m \times n} $ the set of $ m \times n $ matrices with entries in $ \mathbb{F}_q $. Throughout this manuscript, we will fix positive integers $ \ell $, $ n_1, n_2, \ldots, n_\ell $, $ m_1, m_2, \ldots , m_\ell $, and we will consider codes as subsets of $ \prod_{i=1}^\ell \mathbb{F}_q^{m_i \times n_i} $. Linear codes will be $ \mathbb{F}_q $-linear subspaces of such a vector space.

Set $ \mathbf{m} = (m_1, m_2, \ldots, m_\ell) $ and $ \mathbf{n} = (n_1, n_2, \ldots, n_\ell) $. We denote by $ {\rm MC}(\mathbf{m},\mathbf{n}) = \{ (X_i,Y_i)_{i=1}^\ell \mid X_i \subseteq [m_i], Y_i \subseteq [n_i] \} $ the set of multi-covers in $ \prod_{i=1}^\ell \mathbb{F}_q^{m_i \times n_i} $. Finally, given $ X = (X_i, Y_i)_{i=1}^\ell \in {\rm MC}(\mathbf{m},\mathbf{n}) $, we define its size as $ |X| = \sum_{i=1}^\ell (|X_i|+|Y_i|) $ and its projection map $ \pi_X : \prod_{i=1}^\ell \mathbb{F}_q^{m_i \times n_i} \longrightarrow \prod_{i=1}^\ell \mathbb{F}_q^{(m_i - |X_i|) \times (n_i - |Y_i|)} $ by removing from $ C_i \in \mathbb{F}_q^{m_i \times n_i} $ the rows indexed by $ X_i $ and the columns indexed by $ Y_i $ in order to obtain $ \pi_X(C_1, C_2, \ldots, C_\ell) $. Given $ C \in \mathbb{F}_q^{m \times n} $ and $ (C_1, C_2, \ldots, C_\ell) \in \prod_{i=1}^\ell \mathbb{F}_q^{m_i \times n_i} $, we denote by $ C_{a,b} $ and $ C_{i,a,b} $ the entries in row $ a $ and column $ b $ of the matrices $ C $ and $ C_i $, respectively. Throughout the manuscript, we will also assume, without loss of generality, that $ n_i \leq m_i $, for $ i = 1,2, \ldots, \ell $, and $ m_1 \geq \ldots \geq m_\ell $. 

We extend the definition of the cover metric from \cite{gabidulin-lattice, roth} to the multilayer case, as follows. Note that the name ``term rank metric'' is also used for the cover metric, see \cite{gabidulin-survey}.

\begin{definition} \label{def multilayer metric}
Let $ C = (C_1, C_2, \ldots, C_\ell ) \in \prod_{i=1}^\ell \mathbb{F}_q^{m_i \times n_i} $, where $ C_i \in \mathbb{F}_q^{m_i \times n_i} $, for $ i = 1,2, \ldots, \ell $. We say that $ (X_i, Y_i)_{i=1}^\ell \in {\rm MC}(\mathbf{m},\mathbf{n}) $ is a multi-cover of $ C $ if $ (X_i,Y_i) $ is a cover of $ C_i $, which means that if $ C_{i,a,b} \neq 0 $, then $ a \in X_i $ or $ b \in Y_i $, for $ i = 1,2, \ldots, \ell $. We denote by $ {\rm MC}(C) $ the set of all multi-covers of $ C $ (note that $ {\rm MC}(C) \subseteq {\rm MC}(\mathbf{m},\mathbf{n}) $). The multi-cover weight of $ C $ is then defined as $ {\rm wt}_{MC} (C) = \min \{ |X| \mid X \in {\rm MC}(C) \} $. The multi-cover metric is defined as $ {\rm d}_{MC} : (\prod_{i=1}^\ell \mathbb{F}_q^{m_i \times n_i})^2 \longrightarrow \mathbb{N} $, where $ {\rm d}_{MC}(C,D) = {\rm wt}_{MC} (C-D) $, for $ C,D \in \prod_{i=1}^\ell \mathbb{F}_q^{m_i \times n_i} $. Given a code $ \mathcal{C} \subseteq \prod_{i=1}^\ell \mathbb{F}_q^{m_i \times n_i} $, we define its minimum multi-cover distance as $ {\rm d}_{MC}(\mathcal{C}) = \min \{ {\rm d}_{MC}(C,D) \mid C,D \in \mathcal{C}, C \neq D \} $. When considering the minimum distance of a code $ \mathcal{C} $, we implicitly assume that $ |\mathcal{C}| > 1 $.
\end{definition}

Throughout this manuscript, we will consider the column-wise Hamming weight and metric in $ \prod_{i=1}^\ell \mathbb{F}_q^{m_i \times n_i} $, and we will denote them, respectively, as $ {\rm wt}^C_H $ and $ {\rm d}^C_H $. 

Note that the cover metric \cite{gabidulin-lattice, roth} is recovered from the multi-cover metric when $ \ell = 1 $ (in particular, the multi-cover metric is indeed a metric as it is a sum of metrics). Similarly, the multi-cover metric and the Hamming metric coincide in $ \prod_{i=1}^\ell \mathbb{F}_q^{m_i \times n_i} $ if $ m_i = n_i = 1 $, for $ i = 1,2, \ldots, \ell $. In this way, the multi-cover metric interpolates between the cover metric and the Hamming metric. We also have the following trivial relation.

\begin{proposition} \label{prop MC upper bounded H}
For $ C \in \prod_{i=1}^\ell \mathbb{F}_q^{m_i \times n_i} $ and $ \mathcal{C} \subseteq \prod_{i=1}^\ell \mathbb{F}_q^{m_i \times n_i} $, we have
$$ {\rm wt}_{MC}(C) \leq {\rm wt}^C_H(C) \quad \textrm{and} \quad {\rm d}_{MC} (\mathcal{C}) \leq {\rm d}^C_H(\mathcal{C}).  $$
\end{proposition}

As usual, note that if $ \mathcal{C} \subseteq \prod_{i=1}^\ell \mathbb{F}_q^{m_i \times n_i} $ is linear, then
$$ {\rm d}_{MC}(\mathcal{C}) = \min \{ {\rm wt}_{MC}(C) \mid C \in \mathcal{C} \setminus \{ 0 \} \} . $$

Next, we define decoders for multilayer crisscross errors and erasures. This defines implicitly the concept of such errors and erasures. 

\begin{definition} \label{def decoders for crisscross}
Fix positive integers $ \rho $ and $ t $, let $ \mathcal{C} \subseteq \prod_{i=1}^\ell \mathbb{F}_q^{m_i \times n_i} $ and let $ X = (X_i, Y_i)_{i=1}^\ell \in {\rm MC}(\mathbf{m},\mathbf{n}) $ be such that $ |X| = \rho $. A $ t $-error and $ \rho $-erasure-correcting decoder for the code $ \mathcal{C} $ and the erasure pattern $ X $ is a map $ D : \prod_{i=1}^\ell \mathbb{F}_q^{m_i \times n_i} \longrightarrow \prod_{i=1}^\ell \mathbb{F}_q^{(m_i - |X_i|) \times (n_i - |Y_i|)} $ such that $ D(\pi_X(C) + E) = C $, for all $ C \in \mathcal{C} $ and all $ E \in \prod_{i=1}^\ell \mathbb{F}_q^{(m_i - |X_i|) \times (n_i - |Y_i|)} $ with $ {\rm wt}_{MC}(E) \leq t $. In general, the decoder $ D $ will depend on $ \mathcal{C} $, $ X $ and $ t $ (which are known by the decoder), but we do not write this dependency for brevity.
\end{definition} 

We have the following characterization of the multilayer crisscross error and erasure-correction capability of a code. This result implies that the multi-cover metric is the right metric to determine the capability of a code to correct a given number of multilayer crisscross errors and erasures.

\begin{proposition}
Fix positive integers $ t $ and $ \rho $ and fix a code $ \mathcal{C} \subseteq \prod_{i=1}^\ell \mathbb{F}_q^{m_i \times n_i} $. There exists a $ t $-error and $ \rho $-erasure-correcting decoder for $ \mathcal{C} $ and for every $ X \in {\rm MC}(\mathbf{m}, \mathbf{n}) $ with $ |X| = \rho $ if, and only if, $ 2t + \rho < {\rm d}_{MC} (\mathcal{C}) $. 
\end{proposition}
\begin{proof}
First, assume that there exists $ X \in {\rm MC}(\mathbf{m}, \mathbf{n}) $ such that $ |X| = \rho $ and there is no $ t $-error and $ \rho $-erasure-correcting decoder for $ \mathcal{C} $ and $ X $. Hence, there exist distinct $ C,D \in \mathcal{C} $ and $ E,F \in \prod_{i=1}^\ell \mathbb{F}_q^{(m_i - |X_i|) \times (n_i - |Y_i|)} $ such that $ \pi_X(C) + E = \pi_X(D) + F $, $ {\rm wt}_{MC}(E) \leq t $ and $ {\rm wt}_{MC}(F) \leq t $. We have that
$$ {\rm wt}_{MC}(C-D) \leq {\rm wt}_{MC} (\pi_X(C-D)) + |X| = {\rm wt}_{MC} (E-F) + \rho \leq 2t + \rho. $$ 
Hence $ 2t + \rho \geq {\rm d}_{MC} (\mathcal{C}) $.

Conversely, assume that $ 2t + \rho \geq {\rm d}_{MC} (\mathcal{C}) $. Then there exist distinct $ C,D \in \mathcal{C} $ such that $ {\rm wt}_{MC}(C-D) \leq 2t + \rho $. Therefore, there exist $ X, X^\prime, X^{\prime \prime} \in {\rm MC}(\mathbf{m},\mathbf{n}) $ such that $ |X| = \rho $, $ |X^\prime| = |X^{\prime \prime}| = t $ and $ X \cup X^\prime \cup X^{\prime \prime} \in {\rm MC}(C-D) $, where $ X \cup X^\prime \cup X^{\prime \prime} $ is defined in the obvious way. 

Let $ E^\prime, F^\prime \in \prod_{i=1}^\ell \mathbb{F}_q^{m_i \times n_i} $ be formed by the entries of $ -C $ and $ -D $ inside $ X^\prime $ and $ X^{\prime \prime} $ and zero outside, respectively. Finally, let $ E = \pi_X (E^\prime) $ and $ F = \pi_X(F^\prime) $. The reader may verify that $ \pi_X(C) + E = 0 = \pi_X(D) + F $, $ {\rm wt}_{MC}(E) \leq t $ and $ {\rm wt}_{MC}(F) \leq t $. Hence there is no $ t $-error and $ \rho $-erasure-correcting decoder for $ \mathcal{C} $ and $ X $.
\end{proof}

\section{Bounds on code parameters} \label{sec bounds}

In this section, we provide bounds on different multi-cover-metric parameters of a code. We start by providing a Singleton bound for the multi-cover metric. We note that it holds in the same way for the Hamming metric by columns, which constitutes a strengthening of \cite[Th. 3.2]{alberto-fundamental}. However, it may be proven as in that article, thus we omit the proof for brevity. The case of the multi-cover metric is trivial from the case of the Hamming metric by columns by Proposition \ref{prop MC upper bounded H}.

\begin{theorem} \label{th singleton}
Let $ \mathcal{C} \subseteq \prod_{i=1}^\ell \mathbb{F}_q^{m_i \times n_i} $ be a code, and set $ d = {\rm d}_H^C(\mathcal{C}) $ or $ d = {\rm d}_{MC}(\mathcal{C}) $. Let $ \delta $ and $ j $ be the unique integers such that $ d-1 = \sum_{i=1}^{j-1} n_i + \delta $ and $ 0 \leq \delta \leq n_j - 1 $. Then
\begin{equation}
|\mathcal{C}| \leq q^{\sum_{i=j}^\ell m_i n_i - m_j \delta}.
\label{eq singleton bound}
\end{equation}
\end{theorem}

As a consequence, we define maximum multi-cover distance (MMCD) codes as follows.

\begin{definition}
We say that a code $ \mathcal{C} \subseteq \prod_{i=1}^\ell \mathbb{F}_q^{m_i \times n_i} $ is maximum multi-cover distance (MMCD) if equality holds in (\ref{eq singleton bound}).
\end{definition}

In the case of equal numbers of rows, we obtain a more classical form of the Singleton bound from (\ref{eq singleton bound}), as shown in the following proposition. In fact, for equal rows, any bound that holds for the Hamming metric also holds for the multi-cover metric by Proposition \ref{prop MC upper bounded H}. We collect several of these bounds in the following proposition.

\begin{proposition}
Assume that $ m = m_1 = \ldots = m_\ell $ and set $ N = n_1 + n_2 + \cdots + n_\ell $. Let $ \mathcal{C} \subseteq \prod_{i=1}^\ell \mathbb{F}_q^{m \times n_i} $ be a code, and set $ d = {\rm d}_{MC}(\mathcal{C}) $. The following bounds hold.
\begin{enumerate}
\item
Singleton bound:
$$ |\mathcal{C}| \leq q^{m (N - d + 1)} . $$
\item
Hamming bound:
$$ |\mathcal{C}| \leq \left\lfloor \frac{q^{mN}}{ \sum_{j=0}^{\lfloor (d-1)/2 \rfloor} \binom{N}{j} (q^m - 1)^j } \right\rfloor . $$
\item
Plotkin bound: If $ d > (q^m-1)N/q^m $, then
$$ |\mathcal{C}| \leq \left\lfloor \frac{q^{m} d}{ q^m d - (q^m - 1)N } \right\rfloor . $$
\item
Elias bound: For any $ w = 0,1, \ldots , N(q^m-1)/q^m $,
$$ |\mathcal{C}| \leq \left\lfloor \frac{Nd (q^m-1)}{ q^mw^2 + N(d - 2w)(q^m-1) } \cdot \frac{q^{mN}}{ \sum_{j=0}^w \binom{N}{j} (q^m - 1)^j } \right\rfloor . $$
\end{enumerate}
\end{proposition}

Next we explore bounds involving the size of a ball in the multi-cover metric. Define
$$ \mathcal{B}_r(C) = \left\lbrace \left. D \in \prod_{i=1}^\ell \mathbb{F}_q^{m_i \times n_i} \right| {\rm d}_{MC}(C,D) \leq r \right\rbrace , $$
where $ r \in \mathbb{N} $ and $ C \in \prod_{i=1}^\ell \mathbb{F}_q^{m_i \times n_i} $. We also denote $ B_r = |\mathcal{B}_r(0)| $, for $ r \in \mathbb{N} $. The task of finding exactly the size $ B_r $ seems challenging. In Subsection \ref{subsec size ball}, we will provide upper and lower bounds on its size that coincide asymptotically when $ q $ tends to infinity.

We start by a sphere-packing bound, which is straightforward and proven in the same way as for any metric.

\begin{theorem} \label{th sphere-packing bound}
For a code $ \mathcal{C} \subseteq \prod_{i=1}^\ell \mathbb{F}_q^{m_i \times n_i} $, setting $ r = \lfloor ({\rm d}_{MC}(\mathcal{C}) - 1)/2 \rfloor $, we have that
\begin{equation}
B_r \cdot |\mathcal{C}| \leq q^{\sum_{i=1}^\ell m_in_i}. 
\label{eq sphere-packing bound}
\end{equation}
\end{theorem}

As a consequence, we may define perfect codes for the multi-cover metric as follows.

\begin{definition}
We say that a code $ \mathcal{C} \subseteq \prod_{i=1}^\ell \mathbb{F}_q^{m_i \times n_i} $ is a perfect code for the multi-cover metric if equality holds in (\ref{eq sphere-packing bound}).
\end{definition}

We will next compute exactly the value $ B_1 $, which is
\begin{equation}
B_1 = 1 + \sum_{i=1}^\ell \left( n_i \left( q^{m_i} - 1 \right) + m_i \left( q^{n_i} - 1 \right) - m_in_i (q-1) \right).
\label{eq ball size r=1}
\end{equation}
We briefly prove (\ref{eq ball size r=1}). Fix one of the $ \ell $ matrix spaces $ \mathbb{F}_q^{m_i \times n_i} $. To count the number of matrices of cover weight $ 1 $, we may add the number of matrices with only one non-zero column and the number of matrices with only one non-zero row, and then subtract the size of the intersection of these two sets (i.e., the number of matrices with only one non-zero entry). These three numbers are, respectively, $ n_i \left( q^{m_i} - 1 \right) $, $ m_i \left( q^{n_i} - 1 \right) $ and $ m_in_i (q-1) $, hence (\ref{eq ball size r=1}) is proven.

Similarly to \cite[Th. 3.7]{alberto-fundamental}, we may obtain a projective sphere-packing bound by making use of projections. The proof is analogous thus omitted. 

\begin{theorem} \label{th proj sphere packing bound}
Let $ \mathcal{C} \subseteq \prod_{i=1}^\ell \mathbb{F}_q^{m_i \times n_i} $ with $ |\mathcal{C}| \geq 2 $, $ d = {\rm d}_{MC}(\mathcal{C}) $ and $ 3 \leq d \leq n_1 + n_2 + \cdots n_\ell $. Let $ j $ and $ \delta $ be the unique integers such that $ 1 \leq j \leq \ell - 1 $, $ 1 \leq \delta \leq n_{j+1} -1 $, and $ d-3 = \sum_{i=1}^j n_i + \delta $. Define $ n^\prime_{j+1} = n_{j+1} - \delta $ and $ n^\prime_i = n_i $ if $ i \neq j+1 $. Then
$$ |\mathcal{C}| \leq \left\lfloor \frac{q^{\sum_{i=j+1}^\ell m_in^\prime_i}}{1 + \sum_{i=j+1}^\ell \left( n^\prime_{j+1}(q^{m_j} - 1) + m_{j+1}(q^{n^\prime_{j+1}} - 1) - m_{j+1} n^\prime_{j+1}(q-1) \right)} \right\rfloor . $$
\end{theorem}

Combining Theorems \ref{th singleton} and \ref{th proj sphere packing bound}, we may obtain a bound on $ \ell $ for MMCD codes in the case of equal rows and equal columns.

\begin{theorem} \label{th bound on ell for MMCD}
Assume that $ m = m_1 = \ldots = m_\ell $ and $ n = n_1 = \ldots = n_\ell $ (thus $ n \leq m $). Assume that there exists an MMCD code $ \mathcal{C} \subseteq (\mathbb{F}_q^{m \times n})^\ell $ with $ d = {\rm d}_{MC}(\mathcal{C}) $. Let $ 0 \leq \delta \leq n-1 $ be the remainder of the Euclidean division of $ d-3 $ by $ n $. Then
\begin{equation}
\ell \leq \left\lfloor \frac{q^{2m} - 1 - m(q^{n-\delta} - 1) - (n-\delta)(q^m-1) + m(n-\delta)(q-1)}{m(q^n-1) + n (q^m-1) - mn(q-1)} \right\rfloor + \left\lfloor \frac{d-3}{n} \right\rfloor + 1.
\label{eq bound on ell 1}
\end{equation}
Now assume that $ m=n $. If $ q \geq 4 $ and $ n \geq 2 $, or if $ q = 3 $ and $ n \geq 3 $, or if $ q = 2 $ and $ n \geq 4 $, then the upper bound (\ref{eq bound on ell 1}) is tighter than (and thus implies) the bound 
\begin{equation}
\ell \leq \left\lfloor \frac{2q^n}{3n} \right\rfloor + \left\lfloor \frac{d-3}{n} \right\rfloor + 1.
\label{eq bound on ell 2}
\end{equation}
\end{theorem}
\begin{proof}
Define $ \ell^\prime = \ell - \lfloor (d-3)/n \rfloor $ and note that, since $ \mathcal{C} $ is MMCD, then $ |\mathcal{C}| = q^{m (\ell n - d + 1)} $. Thus, applying Theorem \ref{th proj sphere packing bound}, we get
$$ q^{m(\ell n - d + 1)} (1 + (\ell^\prime - 1) (m (q^n-1) + n(q^m-1) -mn (q-1)) + m(q^{n-\delta}-1) $$
$$ + (n-\delta)(q^m-1) -m(n-\delta)(q-1)) \leq q^{m(\ell n - d + 3)}. $$
After simplifying this inequality, we obtain (\ref{eq bound on ell 1}).

Now assume that $ m = n $. Since $ 0 \leq \delta \leq n - 1 $, we have
\begin{equation*}
\begin{split}
\ell - \left\lfloor \frac{d-3}{n} \right\rfloor - 1 & \leq \frac{q^{2n} - 1 - n(q^{n-\delta} - 1) - (n-\delta)(q^n-1) + n(n-\delta)(q-1)}{2n(q^n-1) - n^2(q-1)} \\
& \leq \frac{q^{2n} - 1 - n(q - 1) - (q^n-1) + n^2(q-1)}{2n(q^n-1) - n^2(q-1)}.
\end{split}
\end{equation*}
Now, the inequality
$$ \frac{q^{2n} - 1 - n(q - 1) - (q^n-1) + n^2(q-1)}{2n(q^n-1) - n^2(q-1)} \leq \frac{2q^n}{3n} $$
is equivalent to
$$ q^n \left( \frac{q^n-1}{q-1} - 2n \right) \geq 3 n (n-1). $$
The reader may verify by induction on $ n $ that
$$ q^n \geq 3n(n-1) \quad \textrm{and} \quad \frac{q^n-1}{q-1} - 2n \geq 1, $$
if $ q \geq 4 $ and $ n \geq 2 $, or if $ q = 3 $ and $ n \geq 3 $, or if $ q = 2 $ and $ n \geq 4 $, and the result follows.
\end{proof}

Observe that, in the Hamming-metric case $ m=n=1 $, we have $ \delta = 0 $ and (\ref{eq bound on ell 1}) reads 
\begin{equation}
\ell \leq \frac{q^2 - 1}{2(q-1)-(q-1)} + \frac{d-3}{1} = q + d -2,
\label{eq bound ell for MDS}
\end{equation}
which is a well-known bound on the length of MDS codes \cite[Cor. 7.4.3(ii)]{pless}. Observe now that, by Proposition \ref{prop MC upper bounded H}, the bound (\ref{eq bound ell for MDS}) implies that
\begin{equation}
\ell \leq \left\lfloor \frac{q^n + d - 2}{n} \right\rfloor .
\label{eq bound ell MMCD from MDS}
\end{equation}
However, the bound (\ref{eq bound on ell 2}) is tighter than (\ref{eq bound ell MMCD from MDS}), and therefore, so is (\ref{eq bound on ell 1}), in the cases stated in Theorem \ref{th bound on ell for MMCD}.

We next turn to the non-existence of perfect codes for the multi-cover metric.

\begin{proposition}
Assume that $ n = m_1 = \ldots = m_\ell = n_1 = \ldots = n_\ell $ and $ q $ is even. If $ n $ or $ \ell $ is also even, then there is no perfect code in $ (\mathbb{F}_q^{n \times n})^\ell $ of minimum multi-cover distance $ d = 3 $.
\end{proposition}
\begin{proof}
If there is one such perfect code $ \mathcal{C} \subseteq (\mathbb{F}_q^{n \times n})^\ell $, then according to Theorem \ref{th sphere-packing bound}, we have
$$ q^{\ell n^2} = (1 + \ell n (2(q^n-1) - n(q-1))) |\mathcal{C}| . $$
This means that $ 1 + \ell n (2(q^n-1) - n(q-1)) = q^t $, for some positive integer $ t $, and thus $ 1 + \ell n (2(q^n-1) - n(q-1)) $ is even. However, this is not possible if $ \ell $ or $ n $ is even.
\end{proof}

We conclude with basic Gilbert-Varshamove-like or Sphere-covering existential bound for the multi-cover metric. It is based on ball sizes and is proven in the same way as for any metric.

\begin{theorem}
Let $ d = 1,2, \ldots, N $, where $ N = n_1 + n_2 + \cdots + n_\ell $, and define
$$ k = \left\lceil \log_q \left\lceil \frac{q^{\sum_{i=1}^\ell m_in_i}}{B_{d-1}} \right\rceil \right\rceil . $$
There exists a linear code $ \mathcal{C} \subseteq \prod_{i=1}^\ell \mathbb{F}_q^{m_i \times n_i} $ with $ \dim(\mathcal{C}) = k $ and $ {\rm d}_{MC}(\mathcal{C}) \geq d $.
\end{theorem}

\subsection{The size of a ball} \label{subsec size ball}

In this subsection, we estimate the size of a ball of radius $ r $ in the multi-cover metric, that is, the number $ B_r $. We will provide upper and lower bounds whose limit is the same as that of $ B_r $ as $ q $ tends to infinity. Thus we provide the order of the growth of $ B_r $ as a function of $ q $, being the other parameters fixed. All the proofs of this subsection are deferred to Appendix \ref{app size ball}. Define
$$ \mathcal{S}_r(C) = \left\lbrace \left. D \in \prod_{i=1}^\ell \mathbb{F}_q^{m_i \times n_i} \right| {\rm d}_{MC}(C,D) = r \right\rbrace , $$
where $ r \in \mathbb{N} $ and $ C \in \prod_{i=1}^\ell \mathbb{F}_q^{m_i \times n_i} $. We also denote $ S_r = |\mathcal{S}_r(0)| $, for $ r \in \mathbb{N} $. Clearly, we have $ B_r = \sum_{i=0}^r S_i $, so we only need to compute $ S_r $. For convenience, we define also
$$ S_r^{m,n} = | \{ C \in \mathbb{F}_q^{m \times n} \mid {\rm wt}_{MC}(C) = r \} |, $$
for positive integers $ m $, $ n $ and $ r $. Clearly, it holds that $ S_r = \sum_{\mathbf{r} \in \mathbb{N}^\ell , |\mathbf{r}| = r} \prod_{i=1}^\ell S_{r_i}^{m_i,n_i} $, where $ |\mathbf{r}| = r_1 + r_2 + \cdots + r_\ell $, and therefore,
$$ B_r = \sum_{i=0}^r \sum_{\mathbf{r} \in \mathbb{N}^\ell , |\mathbf{r}| = r} \prod_{i=1}^\ell S_{r_i}^{m_i,n_i}. $$
Thus, we only need to find $ S_r^{m,n} $ for positive integers $ m $, $ n $ and $ r $, which has not yet been computed, to the best of our knowledge.

We start by giving a recursive formula for the sizes of spheres of full radius. We will use such sphere sizes for the general case.

\begin{lemma} \label{lemma size full sphere}
Let $ m $ and $ n $ be positive integers such that $ n \leq m $. Then $ S_1^{m,1} = q^m - 1 $ and 
$$ S_n^{m,n} = q^{mn} - \sum_{i=1}^m \sum_{j=1}^n \binom{m}{i} \binom{n}{j} S_{\min \{ m-i, n-j \}}^{m-i, n-j} . $$
Moreover, we have that $ (q-1)^{mn} \leq S_n^{m,n} \leq q^{mn} $.
\end{lemma}

Therefore, we may recursively compute the numbers $ S_n^{m,n} $, where $ n \leq m $. Using these numbers, we may give the following upper and lower bounds for $ S_r^{m,n} $, for any $ r = 1,2, \ldots , \min \{ m,n \} $.

\begin{theorem} \label{th size ball}
For any positive integers $ m $, $ n $ and $ r $, with $ 1 \leq r \leq \min \{ m , n \} $, it holds that $ S^{m,n}_r = {\rm UB}_r - {\rm DC}_r $, where $ {\rm UB}_r \geq {\rm DC}_r \geq 0 $ are integers such that
\begin{equation*}
\begin{split}
{\rm UB}_r & = \sum_{s=0}^r \binom{m}{s} \binom{n}{r-s} (q^{m-s} - 1)^{r-s} (q^{n-r+s}-1)^s q^{s(r-s)}, \\
{\rm DC}_r & \leq \sum_{\omega = 0}^{r-1} \sum_{u=0}^\omega \binom{m}{u} \binom{n}{\omega - u} \sum_{s=0}^{r-\omega}\sum_{t=s}^{r-\omega} \left[ \binom{m-u}{s,t} \binom{n-\omega + u}{r-\omega-s, r - \omega - t} \right. \\
& \phantom{- \sum_{\omega = 0}^{r-1} \sum_{u=0}^\omega} \cdot q^{u(\omega - u)} \cdot S^{t, r-\omega-s}_{\min \{ t, r-\omega-s \}} \cdot S^{s,r-\omega-t}_{\min \{ s, r-\omega-t \}} \\
& \phantom{- \sum_{\omega = 0}^{r-1} \sum_{u=0}^\omega} \cdot \left( q^{t+s}(q^{m-u-s-t} - 1) + (q^s-1) (q^t-1) \right)^{\omega - u} \\
& \phantom{- \sum_{\omega = 0}^{r-1} \sum_{u=0}^\omega} \cdot \bigg( q^{(r-\omega-t)+(r-\omega-s)}(q^{n-(\omega-u)-(r-\omega-t)-(r-\omega-s)} - 1) \\
& \phantom{- \sum_{\omega = 0}^{r-1} \sum_{u=0}^\omega \cdot \bigg(} + (q^{r-\omega-s}-1)(q^{r-\omega-t} - 1) \bigg)^u \bigg] , \\
{\rm DC}_r & \geq \sum_{\omega = 0}^{r-1} \sum_{u=0}^\omega \binom{m}{u} \binom{n}{\omega - u} \sum_{s=0}^{r-\omega}\sum_{t=s}^{r-\omega} \left[ \binom{m-u}{s,t} \binom{n-\omega + u}{r-\omega-s, r - \omega - t} \right. \\
& \phantom{- \sum_{\omega = 0}^{r-1} \sum_{u=0}^\omega} \cdot q^{u(\omega - u)} \cdot (q-1)^{t (r-\omega-s) + s (r-\omega-t)} \\
& \phantom{- \sum_{\omega = 0}^{r-1} \sum_{u=0}^\omega} \cdot \left( q^{t+s}(q^{m-u-s-t} - 1) + (q^s-1) (q^t-1) \right)^{\omega - u} \\
& \phantom{- \sum_{\omega = 0}^{r-1} \sum_{u=0}^\omega} \cdot \bigg( q^{(r-\omega-t)+(r-\omega-s)}(q^{n-(\omega-u)-(r-\omega-t)-(r-\omega-s)} - 1) \\
& \phantom{- \sum_{\omega = 0}^{r-1} \sum_{u=0}^\omega \cdot \bigg(} + (q^{r-\omega-s}-1)(q^{r-\omega-t} - 1) \bigg)^u \bigg] .
\end{split}
\end{equation*}
Here, for integers $ m \geq 1 $, $ 0 \leq s \leq m $ and $ s \leq t \leq m $, we define
$$ \binom{m}{s,t} = \left\lbrace \begin{array}{ll}
\frac{m!}{s! t! (m-s-t)!} & \textrm{if } s \neq 0 \textrm{ or } t \neq 0, \\
0 & \textrm{if } s = t = 0 .
\end{array} \right. $$
\end{theorem}

In the case $ r = 1 $, the reader may verify that $ {\rm UB}_1 = n(q^m-1) + m(q^n-1) $ and the upper and lower bounds in Theorem \ref{th size ball} coincide and give $ {\rm DC}_1 = mn (q-1) $. Hence $ S^{m,n}_1 = n(q^m-1) + m(q^n-1) - mn (q-1) $, as computed in (\ref{eq ball size r=1}).

In general, for $ r > 1 $, the upper and lower bounds in Theorem \ref{th size ball} do not coincide. However, when $ m $, $ n $ and $ r $ are fixed but $ q $ grows, the upper and lower bounds in Theorem \ref{th size ball} coincide asymptotically. In particular, we may asymptotically compute the order of both $ S^{m,n}_r $ and $ B^{m,n}_r = \sum_{i=0}^r S^{m,n}_i $ as functions of $ q $.

\begin{corollary} \label{cor order of ball sizes}
Let $ 1 \leq r \leq n \leq m $ be positive integers. Then 
$$ \lim_{q \rightarrow \infty} \binom{n}{r} \frac{q^{mr}}{S^{m,n}_r} = \lim_{q \rightarrow \infty} \binom{n}{r} \frac{q^{mr}}{B^{m,n}_r} = \lim_{q \rightarrow \infty}\frac{S^{m,n}_r}{B^{m,n}_r} = 1. $$ 
\end{corollary}

\section{Duality, puncturing and shortening} \label{sec duality, puncturing and shortening}

In this section, we study duality, puncturing and shortening for the multi-cover metric. We will consider duality with respect to the trace product, given by
\begin{equation}
\langle C,D \rangle = \sum_{i=1}^\ell {\rm Tr}(C_iD_i),
\label{eq def trace product}
\end{equation}
where $ C = (C_1, C_2, \ldots, C_\ell), D = (D_1, D_2, \ldots, D_\ell) \in \prod_{i=1}^\ell \mathbb{F}_q^{m_i \times n_i} $, and where $ {\rm Tr} (A) $ denotes the trace of the matrix $ A $. Observe that $ \langle \cdot , \cdot \rangle $ is nothing but the usual inner product over $ \mathbb{F}_q $ seeing $ \prod_{i=1}^\ell \mathbb{F}_q^{m_i \times n_i} $ as $ \mathbb{F}_q^{\sum_{i=1}^\ell m_in_i} $. We then define the dual of a linear code $ \mathcal{C} \subseteq \prod_{i=1}^\ell \mathbb{F}_q^{m_i \times n_i} $ in the usual way,
$$ \mathcal{C}^\perp = \left\lbrace \left. D \in \prod_{i=1} \mathbb{F}_q^{m_i \times n_i} \right| \langle C,D \rangle = 0, \textrm{ for all } C \in \mathcal{C} \right\rbrace . $$

\subsection{Puncturing and shortening} \label{subsec puncturing shortening}

In this subsection, we define puncturing and shortening, which extend the classical concepts of puncturing and shortening for the Hamming metric \cite[Sec. 1.5]{pless}. As in the classical Hamming-metric case, puncturing and shortening enables us to explicitly construct shorter codes from known codes. To the best of our knowledge, these concepts have not yet been introduced in the classical cover metric case ($ \ell = 1 $).

\begin{definition}
Let $ X = (X_i, Y_i)_{i=1}^\ell \in {\rm MC}(\mathbf{m},\mathbf{n}) $. Given a code $ \mathcal{C} \subseteq \prod_{i=1}^\ell \mathbb{F}_q^{m_i \times n_i} $, we define its puncturing and shortening on $ X $, respectively, as
\begin{equation*}
\begin{split}
\mathcal{C}_X & = \pi_X(\mathcal{C}), \\
\mathcal{C}^X & = \{ \pi_X(C) \mid C \in \mathcal{C}, C_{i,a,b} = 0 \textrm{ if } a \in X_i \textrm{ or } b \in Y_i, 1 \leq i \leq \ell \} ,
\end{split}
\end{equation*}
both of which are codes in $ \prod_{i=1}^\ell \mathbb{F}_q^{(m_i - |X_i|) \times (n_i - |Y_i|)} $. 
\end{definition}

We now describe the basic properties of punctured and shortened codes in general.

\begin{proposition} \label{prop puncturing shortening properties}
Given a linear code $ \mathcal{C} \subseteq \prod_{i=1}^\ell \mathbb{F}_q^{m_i \times n_i} $ and a multi-cover $ X = (X_i, Y_i)_{i=1}^\ell \in {\rm MC}(\mathbf{m}, \mathbf{n}) $, the following hold:
\begin{enumerate}
\item
$ (\mathcal{C}_X)^\perp = (\mathcal{C}^\perp)^X $ and $ (\mathcal{C}^X)^\perp = (\mathcal{C}^\perp)_X $.
\item
$ \dim(\mathcal{C}_X) \geq \dim(\mathcal{C}) - \sum_{i=1}^\ell \left( n_i |X_i| + m_i |Y_i| - |X_i| \cdot |Y_i| \right) $.
\item
$ \dim(\mathcal{C}^X) \geq \dim(\mathcal{C}) - \sum_{i=1}^\ell \left( n_i |X_i| + m_i |Y_i| - |X_i| \cdot |Y_i| \right) $.
\item
$ {\rm d}_{MC} (\mathcal{C}_X) \geq {\rm d}_{MC}(\mathcal{C}) - |X| $.
\item
$ {\rm d}_{MC} (\mathcal{C}^X) \geq {\rm d}_{MC}(\mathcal{C}) $.
\end{enumerate}
\end{proposition}
\begin{proof}
Items 4 and 5 are straightforward. Item 1 coincides with the classical result \cite[Th. 1.5.7(i)]{pless}. Items 2 and 3 coincide with the classical results after realizing that the number of deleted positions is exactly $ \sum_{i=1}^\ell \left( n_i |X_i| + m_i |Y_i| - |X_i| \cdot |Y_i| \right) $. 
\end{proof}

More interestingly, we may obtain shorter linear MMCD codes from known linear MMCD codes.

\begin{theorem} \label{th MMCD puncturing and shortening 1}
Let $ \mathcal{C} \subseteq \prod_{i=1}^\ell \mathbb{F}_q^{m_i \times n_i} $ be a linear MMCD code. Set $ d = {\rm d}_{MC}(\mathcal{C}) $ and let $ \delta $ and $ j $ be the unique integers such that $ d-1 = \sum_{i=1}^{j-1}n_i +\delta $ and $ 0 \leq \delta \leq n_j - 1 $. Let $ X = (X_i, Y_i)_{i=1}^\ell \in {\rm MC}(\mathbf{m}, \mathbf{n}) $. The following hold:
\begin{enumerate}
\item
Let $ 1 \leq k \leq j $, assume that $ d > 1 $, $ X_i = \varnothing $, for all $ i = 1,2, \ldots, \ell $, $ Y_i = \varnothing $ if $ i \neq k $ and $ |Y_k| = 1 $. Further assume $ \delta > 0 $ if $ k=j $. Then $ \mathcal{C}_X $ is a linear MMCD code with $ \dim(\mathcal{C}_X) = \dim(\mathcal{C}) $ and $ {\rm d}_{MC}(\mathcal{C}_X) = {\rm d}_{MC}(\mathcal{C}) - 1 $.
\item
Let $ j+1 \leq k \leq \ell $, and assume that $ Y_i = \varnothing $, for all $ i = 1,2, \ldots, \ell $, $ X_i = \varnothing $ if $ i \neq k $ and $ |X_k| = 1 $. Then $ \mathcal{C}^X $ is a linear MMCD code with $ \dim(\mathcal{C}^X) = \dim(\mathcal{C}) - n_k $ and $ {\rm d}_{MC}(\mathcal{C}^X) = {\rm d}_{MC}(\mathcal{C}) $.
\item
Let $ j \leq k \leq \ell $, and assume that $ X_i = \varnothing $, for all $ i = 1,2, \ldots, \ell $, $ Y_i = \varnothing $ if $ i \neq k $ and $ |Y_k| = 1 $. Then $ \mathcal{C}^X $ is a linear MMCD code with $ \dim(\mathcal{C}^X) = \dim(\mathcal{C}) - m_k $ and $ {\rm d}_{MC}(\mathcal{C}^X) = {\rm d}_{MC}(\mathcal{C}) $.
\end{enumerate}
\end{theorem}
\begin{proof}
We start by proving Item 1. Let $ n^\prime_i = n_i $ if $ i \neq k $, and let $ n^\prime_k = n_k - 1 $. Note that $ \mathcal{C}_X \subseteq \prod_{i=1}^\ell \mathbb{F}_q^{m_i \times n^\prime_i} $. Set also $ \delta^\prime = \delta - 1 $ if $ k = j $, and $ \delta^\prime = \delta $ otherwise. Thus $ 0 \leq \delta^\prime \leq n^\prime_j - 1 $ by the assumptions. By Proposition \ref{prop puncturing shortening properties}, we have that
$$ {\rm d}_{MC}(\mathcal{C}_X) - 1 \geq d-2 = \sum_{i=1}^{j-1}n_i + \delta - 1 = \sum_{i=1}^{j-1} n^\prime_i +\delta^\prime, $$
and since $ d > 1 $, we also have that 
$$ \dim(\mathcal{C}_X) = \dim(\mathcal{C}) = \sum_{i=j}^\ell m_in_i - m_j \delta = \sum_{i=j}^\ell m_in^\prime_i - m_j \delta^\prime . $$
Therefore, $ \mathcal{C}_X $ must be MMCD and the inequality above is an equality.

We next prove Item 2. Let $ m^\prime_i = m_i $ if $ i \neq k $, and let $ m^\prime_k = m_k - 1 $. Note that $ \mathcal{C}^X \subseteq \prod_{i=1}^\ell \mathbb{F}_q^{m^\prime_i \times n_i} $. By Proposition \ref{prop puncturing shortening properties}, we have that
\begin{equation*}
\begin{split}
{\rm d}_{MC}(\mathcal{C}^X) - 1 & \geq d - 1 = \sum_{i=1}^{j-1}n_i + \delta , \textrm{ and} \\
\dim(\mathcal{C}^X) & \geq \dim(\mathcal{C}) - n_k = \sum_{i=j}^\ell m^\prime_i n_i - m^\prime_j \delta .
\end{split}
\end{equation*}
Therefore, $ \mathcal{C}^X $ must be MMCD and the inequalities above are equalities.

Finally we prove Item 3. Let $ n^\prime_i = n_i $ if $ i \neq k $, and let $ n^\prime_k = n_k - 1 $. Note that $ \mathcal{C}^X \subseteq \prod_{i=1}^\ell \mathbb{F}_q^{m_i \times n^\prime_i} $. By Proposition \ref{prop puncturing shortening properties}, we have that
\begin{equation*}
\begin{split}
{\rm d}_{MC}(\mathcal{C}^X) - 1 & \geq d - 1 = \sum_{i=1}^{j-1}n_i + \delta = \sum_{i=1}^{j-1}n^\prime_i + \delta , \textrm{ and} \\
\dim(\mathcal{C}^X) & \geq \dim(\mathcal{C}) - m_k = \sum_{i=j}^\ell m_i n^\prime_i - m_j \delta .
\end{split}
\end{equation*}
Therefore, $ \mathcal{C}^X $ must be MMCD and the inequalities above are equalities.
\end{proof}

\subsection{Dually MMCD codes and MDS codes by rows and columns} \label{subsec dually mmcd}

In Section \ref{sec definitions}, we considered the Hamming metric by columns and related it to the multi-cover metric. We may similarly define the Hamming metric by rows in $ \mathbb{F}_q^{n \times n} $ by transposition, i.e., we may consider
$$ \mathcal{C}^\intercal = \{ C^\intercal \mid C \in \mathcal{C} \} \subseteq \mathbb{F}_q^{n \times n}, $$
for $ \mathcal{C} \subseteq \mathbb{F}_q^{n \times n} $, where $ C^\intercal $ denotes the transposed of a matrix $ C $. We may define the minimum Hamming distance of $ \mathcal{C} $ by rows as $ {\rm d}_H^R(\mathcal{C}) = {\rm d}_H^C(\mathcal{C}^\intercal) $. The Singleton bound from Theorem \ref{th singleton} holds in the same way for both rows and columns, i.e., 
$$ |\mathcal{C}| \leq q^{n (n- d + 1)}, $$
whether $ d = {\rm d}_H^R $ or $ d = {\rm d}_H^C $. A code attaining this bound for $ {\rm d}_H^R $ will be called MDS by rows, and analogously for columns. Clearly, if $ \mathcal{C} $ is MMCD, then it is both MDS by rows and by columns. In particular, if it is linear, then $ \mathcal{C}^\perp $ is also MDS by rows and by columns, since the MDS property is preserved by duality (the proof of \cite[Th. 2.4.3]{pless} may be extended to matrix codes and the trace product in a straightforward way). 

However, as we now show, the dual of a linear MMCD code is not necessarily MMCD itself and a linear code that is MDS by rows and by columns is not necessarily MMCD either.

\begin{example}
Consider $ \mathcal{C} \subseteq \mathbb{F}_2^{3 \times 3} $ generated by
$$ A = \left( \begin{array}{ccc}
1 & 1 & 1 \\
1 & 0 & 0 \\
1 & 0 & 0 
\end{array} \right), \quad B = \left( \begin{array}{ccc}
0 & 1 & 0 \\
1 & 1 & 1 \\
0 & 1 & 0
\end{array} \right), \quad \textrm{and} \quad C = \left( \begin{array}{ccc}
0 & 0 & 1 \\
0 & 0 & 1 \\
1 & 1 & 1
\end{array} \right). $$
Since $ \mathcal{C} = \{ 0 , A, B, C, A+B, B+C, C+A, A+B+C \} $ and $ \dim(\mathcal{C}) = 3 $, it is easy to see that $ \mathcal{C} $ is MDS by columns and by rows since $ {\rm d}_H^R(\mathcal{C}) = {\rm d}_H^C(\mathcal{C}) = 3 $, but it is not MMCD, since $ {\rm d}_{MC}(\mathcal{C}) = 2 $. Now, $ \mathcal{C}^\perp \subseteq \mathbb{F}_2^{3 \times 3} $ has $ \dim(\mathcal{C}^\perp) = 6 $. Moreover, by inspection one can see that there is no $ D \in \mathcal{C}^\perp $ with $ {\rm wt}_{MC}(D) = 1 $. Therefore, $ {\rm d}_{MC}(\mathcal{C}^\perp) = 2 $ and $ \mathcal{C}^\perp $ is a linear MMCD code, even though $ \mathcal{C} $ is not.
\end{example}

The example above motivates the following definition.

\begin{definition}
We say that $ \mathcal{C} \subseteq \prod_{i=1}^\ell \mathbb{F}_q^{m_i \times n_i} $ is a dually MMCD code if it is linear and both $ \mathcal{C} $ and $ \mathcal{C}^\perp $ are MMCD codes.
\end{definition}

In Section \ref{sec constructions}, we will provide some explicit constructions of dually MMCD codes for general parameters.

To conclude the subsection, we observe that the equivalence of linear MMCD codes, dually MMCD codes and MDS codes by rows and by columns holds for very small parameters. For the case $ \ell > 1 $, we need to consider different combinations of transpositions in different positions. To this end, given $ C = (C_1, C_2, \ldots, C_\ell) \in \prod_{i=1}^\ell \mathbb{F}_q^{m_i \times n_i} $ and $ \mathbf{t} \in \{ 0,1\}^\ell $, we define $ C^\mathbf{t} = (D_1, D_2, \ldots, D_\ell) \in \prod_{i=1}^\ell \mathbb{F}_q^{n_i \times m_i} $, where  
$$ D_i = \left\lbrace \begin{array}{ll}
C_i & \textrm{if } t_i = 0, \\
C_i^\intercal & \textrm{if } t_i = 1.
\end{array} \right. $$
We then define $ \mathcal{C}^\mathbf{t} = \{ C^\mathbf{t} \mid C \in \mathcal{C} \} \subseteq \prod_{i=1}^\ell \mathbb{F}_q^{n_i \times m_i} $, for $ \mathcal{C} \subseteq \prod_{i=1}^\ell \mathbb{F}_q^{m_i \times n_i} $.

The proofs of the following two propositions are based on the fact that, in the two cases, we only need to consider covers in $ \mathbb{F}_q^{m \times n} $ of sizes $ 1 $ or $ 2 $, and these can always be chosen as only columns or only rows in such cases.

\begin{proposition}
Let $ \mathcal{C} \subseteq (\mathbb{F}_q^{n \times n})^\ell $ be a linear code with $ \dim(\mathcal{C}) = n (\ell n -1) $. Then $ \mathcal{C} $ is MMCD if, and only if, $ \mathcal{C}^\mathbf{t} $ is MDS by columns for all $ \mathbf{t} \in \{ 0,1 \}^\ell $.
\end{proposition}

\begin{proposition}
Let $ \mathcal{C} \subseteq (\mathbb{F}_q^{2 \times 2})^\ell $ be a linear code. The following are equivalent:
\begin{enumerate}
\item
$ \mathcal{C}^\mathbf{t} $ is MDS by columns for all $ \mathbf{t} \in \{ 0,1 \}^\ell $.
\item
$ \mathcal{C} $ is MMCD.
\item
$ \mathcal{C} $ is dually MMCD.
\end{enumerate}
\end{proposition}

Recall that the multi-cover metric in $ (\mathbb{F}_q^{1 \times 1})^\ell $ is simply the classical Hamming metric in $ \mathbb{F}_q^\ell $, hence the previous proposition also holds but is trivial in this case.

\subsection{Information multi-covers} \label{subsec information cover sets}

In this subsection, we explore the notion of information sets for the multi-cover metric. We will present two types, information multi-covers and complementary information multi-covers. The former characterize when the dual code is MMCD and the latter characterize when the primary code is MMCD. Both types of multi-covers extend information sets in the Hamming-metric case \cite[p. 4]{pless} and our characterizations recover in that case the characterizations of classical MDS codes based on information sets \cite[Th. 2.4.3]{pless}.

We start with a notion of support space similar to that from the Hamming-metric case \cite[Sec. II]{forney}. 

\begin{definition}
We define the support space of $ X \in {\rm MC}(\mathbf{m},\mathbf{n}) $ as the vector subspace
$$ \mathcal{V}_X = \left\lbrace \left. C \in \prod_{i=1}^\ell \mathbb{F}_q^{m_i \times n_i} \right| X \in {\rm MC}(C) \right\rbrace. $$
\end{definition}

Observe that, if $ X = (X_i,Y_i)_{i=1}^\ell \in {\rm MC}(\mathbf{m},\mathbf{n}) $, then
$$ \dim(\mathcal{V}_X) = \sum_{i=1}^\ell (n_i|X_i| + m_i|Y_i| - |X_i|\cdot |Y_i|) . $$
Other basic properties hold, as in the Hamming-metric case. However, duals of support spaces will not be themselves support spaces in general. In fact, we may easily characterize when duals of support spaces are again support spaces. The proof is straightforward.

\begin{proposition}
Given $ X \in {\rm MC}(\mathbf{m},\mathbf{n}) $, there exist $ X^\prime \in {\rm MC}(\mathbf{m},\mathbf{n}) $ with $ \mathcal{V}_X^\perp = \mathcal{V}_{X^\prime} $ if, and only if, for all $ i = 1,2, \ldots, \ell $, we have $ X_i = \varnothing $ or $ Y_i = \varnothing $.
\end{proposition}

We next define information multi-covers and complementary information multi-covers. Recall that, for $ X = (X_i,Y_i)_{i=1}^\ell \in {\rm MC}(\mathbf{m},\mathbf{n}) $, we have defined the projection map
\begin{equation}
\pi_X : \prod_{i=1}^\ell \mathbb{F}_q^{m_i \times n_i} \longrightarrow \prod_{i=1}^\ell \mathbb{F}_q^{(m_i - |X_i|) \times (n_i - |Y_i|)}
\label{eq def projection map}
\end{equation}
by removing all the entries covered by $ X $ (i.e., all rows indexed by $ X_i $ and all columns indexed by $ Y_i $, for $ i = 1,2, \ldots, \ell $). We then define complementary information multi-covers as follows.

\begin{definition}
Given a code $ \mathcal{C} \subseteq \prod_{i=1}^\ell \mathbb{F}_q^{m_i \times n_i} $, we say that $ X = (X_i,Y_i)_{i=1}^\ell \in {\rm MC}(\mathbf{m},\mathbf{n}) $ is a complementary information multi-cover of $ \mathcal{C} $ if one of the following equivalent conditions hold:
\begin{enumerate}
\item
The restriction of the projection map in (\ref{eq def projection map}) to the code $ \mathcal{C} $, that is, $ \pi_X : \mathcal{C} \longrightarrow \prod_{i=1}^\ell \mathbb{F}_q^{(m_i - |X_i|) \times (n_i - |Y_i|)} $, is injective.
\item
$ |\mathcal{C}_X| = |\mathcal{C}| $.
\item
(If $ \mathcal{C} $ is linear) $ \dim(\mathcal{C}_X) = \dim(\mathcal{C}) $.
\end{enumerate}
\end{definition}

In other words, when removing the entries indexed by $ X $ from a codeword in $ \mathcal{C} $, we may still reconstruct the whole codeword. Since $ X $ is formed by the indices of the erased coordinates (i.e., the indices not containing the information that we have), the term \textit{complementary} information cover is justified.

We next observe the monotonicity of complementary information multi-covers, which will allow us to only consider maximal complementary information multi-covers in the subsequent results.

\begin{lemma}
If $ \mathcal{C} \subseteq \prod_{i=1}^\ell \mathbb{F}_q^{m_i \times n_i} $ is a code and $ X = (X_i,Y_i)_{i=1}^\ell \in {\rm MC}(\mathbf{m},\mathbf{n}) $ is a complementary information multi-cover of $ \mathcal{C} $, then so is any $ X^\prime = (X^\prime_i,Y^\prime_i)_{i=1}^\ell \in {\rm MC}(\mathbf{m},\mathbf{n}) $ such that $ X^\prime_i \subseteq X_i $ and $ Y^\prime_i \subseteq Y_i $, for $ i = 1,2, \ldots, \ell $.
\end{lemma}
 
As in the classical Hamming-metric case \cite[Th. 1.4.15]{pless}, it is easy to characterize the minimum multi-cover distance in terms of support spaces and complementary information multi-covers. The proof is straightforward.

\begin{lemma} \label{lemma charact min dist from supports}
Given a code $ \mathcal{C} \subseteq \prod_{i=1}^\ell \mathbb{F}_q^{m_i \times n_i} $ and a positive integer $ d $, the following are equivalent:
\begin{enumerate}
\item
$ {\rm d}_{MC}(\mathcal{C}) \geq d $.
\item
$ |\mathcal{C} \cap \mathcal{V}_X | = 1 $, for all $ X \in {\rm MC}(\mathbf{m}, \mathbf{n}) $ such that $ |X| = d-1 $. 
\item
Any $ X \in {\rm MC}(\mathbf{m},\mathbf{n}) $ such that $ |X| = d-1 $ is a complementary information multi-cover of $ \mathcal{C} $.
\end{enumerate}
\end{lemma}

As a consequence, we may characterize linear MMCD codes in terms of complementary information multi-covers, as in the Hamming-metric case \cite[Th. 2.4.3]{pless}.

\begin{proposition} \label{prop MMCD information covers}
Assume that $ m = m_1 = \ldots = m_\ell $ and set $ N = n_1 + n_2 + \cdots + n_\ell $. Given a linear code $ \mathcal{C} \subseteq \prod_{i=1}^\ell \mathbb{F}_q^{m \times n_i} $ whose dimension is a multiple of $ m $, the following are equivalent:
\begin{enumerate}
\item
$ \mathcal{C} $ is MMCD.
\item
Any $ X \in {\rm MC}(\mathbf{m},\mathbf{n}) $ with $ m(N - |X|) = \dim(\mathcal{C}) $ is a complementary information multi-cover of $ \mathcal{C} $.
\end{enumerate}
\end{proposition}
\begin{proof}
Assume first that Item 1 holds. Then $ d = N - \dim(\mathcal{C})/m + 1 $, hence $ m(N - |X|) = \dim(\mathcal{C}) $ if, and only if, $ |X| = d-1 $, and in that case, $ X $ is a complementary information multi-cover by Lemma \ref{lemma charact min dist from supports}. Assume now that Item 2 holds. Define $ d = N - \dim(\mathcal{C})/m + 1 $, which is an integer since $ \dim(\mathcal{C}) $ is a multiple of $ m $. Now, $ m(N - |X|) = \dim(\mathcal{C}) $ if, and only if, $ |X| = d-1 $. Since all such multi-covers are complementary information multi-covers, then $ {\rm d}_{MC}(\mathcal{C}) \geq d $ by Lemma \ref{lemma charact min dist from supports}, and by Theorem \ref{th singleton}, $ \mathcal{C} $ must be MMCD.
\end{proof}

As noted above, duality does not work as expected for all multi-covers $ X \in {\rm MC}(\mathbf{m}, \mathbf{n}) $. Thus we define a notion of information multi-cover, which is a dual concept of that of complementary information multi-covers. These are the type of multi-covers that characterize when the dual of a linear code is MMCD. Since we are concerned with duality, we only consider linear codes in the remainder of this subsection.

For $ X \in {\rm MC}(\mathbf{m},\mathbf{n}) $, define the projection map
\begin{equation}
\pi^X : \prod_{i=1}^\ell \mathbb{F}_q^{m_i \times n_i} \longrightarrow \mathbb{F}_q^{\sum_{i=1}^\ell (n_i|X_i| + m_i|Y_i| - |X_i| \cdot |Y_i|)}
\label{eq def projection map dual}
\end{equation}
defined by removing all the entries not covered by $ X $.

\begin{definition}
Given a linear code $ \mathcal{C} \subseteq \prod_{i=1}^\ell \mathbb{F}_q^{m_i \times n_i} $, we say that $ X = (X_i,Y_i)_{i=1}^\ell \in {\rm MC}(\mathbf{m},\mathbf{n}) $ is an information multi-cover of $ \mathcal{C} $ if one of the following equivalent conditions hold:
\begin{enumerate}
\item
The restriction of the projection map in (\ref{eq def projection map dual}) to the code $ \mathcal{C} $, that is, $ \pi^X : \mathcal{C} \longrightarrow \mathbb{F}_q^{\sum_{i=1}^\ell (n_i|X_i| + m_i|Y_i| - |X_i| \cdot |Y_i|)} $, is surjective.
\item
$ \dim(\mathcal{C}^X) = \dim(\mathcal{C}) - \sum_{i=1}^\ell (m|X_i| + m|Y_i| - |X_i| \cdot |Y_i|) $.
\end{enumerate}
\end{definition}

Once again, it can be easily shown that information multi-covers satisfy the following monotonicity property. Again, this monotonicity property will allow us to only consider maximal information multi-covers in the subsequent results.

\begin{lemma}
If $ \mathcal{C} \subseteq \prod_{i=1}^\ell \mathbb{F}_q^{m_i \times n_i} $ is a linear code and $ X = (X_i,Y_i)_{i=1}^\ell \in {\rm MC}(\mathbf{m},\mathbf{n}) $ is an information multi-cover of $ \mathcal{C} $, then so is any $ X^\prime = (X^\prime_i,Y^\prime_i)_{i=1}^\ell \in {\rm MC}(\mathbf{m},\mathbf{n}) $ such that $ X^\prime_i \subseteq X_i $ and $ Y^\prime_i \subseteq Y_i $, for $ i = 1,2, \ldots, \ell $.
\end{lemma}

The reason why we use the term information multi-covers is due to the following theorem, which acts as the dual of Proposition \ref{prop MMCD information covers}

\begin{proposition} \label{prop MMCD information covers dual}
Assume that $ m = m_1 = \ldots = m_\ell $ and set $ N = n_1 + n_2 + \cdots + n_\ell $. Given a linear code $ \mathcal{C} \subseteq \prod_{i=1}^\ell \mathbb{F}_q^{m \times n_i} $ whose dimension is a multiple of $ m $, the following are equivalent:
\begin{enumerate}
\item
$ \mathcal{C}^\perp $ is MMCD.
\item
Any $ X \in {\rm MC}(\mathbf{m},\mathbf{n}) $ with $ m|X| = \dim(\mathcal{C}) $ is an information multi-cover of $ \mathcal{C} $.
\end{enumerate}
\end{proposition}
\begin{proof}
Let $ d^\perp = {\rm d}_{MC}(\mathcal{C}^\perp) $. By definition, $ \mathcal{C}^\perp $ is MMCD if, and only if, $ \dim(\mathcal{C}^\perp) = m (N - d^\perp + 1) $, which is in turn equivalent to $ \dim(\mathcal{C}) = m(d^\perp - 1) $. Hence by Lemma \ref{lemma charact min dist from supports}, $ \mathcal{C}^\perp $ is MMCD if, and only if, for any $ X = (X_i, Y_i)_{i=1}^\ell \in {\rm MC}(\mathbf{m}, \mathbf{n}) $ such that $ m|X| = \dim(\mathcal{C}) $, it holds that $ \dim(\mathcal{C}^\perp \cap \mathcal{V}_X) = 0 $. Now, by linear algebra, we have that
\begin{equation*}
\begin{split}
\dim(\mathcal{C}^\perp \cap \mathcal{V}_X) & = mN - \dim(\mathcal{C} + \mathcal{V}_X^\perp) \\
& = mN - \dim(\mathcal{C}) - \dim(\mathcal{V}_X^\perp) + \dim(\mathcal{C} \cap \mathcal{V}_X^\perp) \\
& = \dim(\mathcal{C} \cap \mathcal{V}_X^\perp)  - \dim(\mathcal{C}) + \sum_{i=1}^\ell (m|X_i| + m|Y_i| - |X_i| \cdot |Y_i|) .
\end{split}
\end{equation*}
Furthermore, note that $ \pi_X(\mathcal{C} \cap \mathcal{V}_X^\perp) = \mathcal{C}^X $. Hence, $ \dim(\mathcal{C}^\perp \cap \mathcal{V}_X) = 0 $ if, and only if,
$$ \dim(\mathcal{C}^X) = \dim(\mathcal{C}) - \sum_{i=1}^\ell (m|X_i| + m|Y_i| - |X_i| \cdot |Y_i|) , $$
which means that $ X $ is an information multi-cover of $ \mathcal{C} $, and we are done.
\end{proof}

Hence we conclude the following.

\begin{corollary}
Assume that $ m = m_1 = \ldots = m_\ell $ and set $ N = n_1 + n_2 + \cdots + n_\ell $. Given a linear code $ \mathcal{C} \subseteq \prod_{i=1}^\ell \mathbb{F}_q^{m \times n_i} $ whose dimension is a multiple of $ m $, the following are equivalent:
\begin{enumerate}
\item
$ \mathcal{C}$ is dually MMCD.
\item
Any $ X \in {\rm MC}(\mathbf{m},\mathbf{n}) $ with $ m|X| = \dim(\mathcal{C}) $ is an information multi-cover of $ \mathcal{C} $ and any $ Y \in {\rm MC}(\mathbf{m},\mathbf{n}) $ with $ m (N-|Y|) = \dim(\mathcal{C}) $ is a complementary information multi-cover of $ \mathcal{C} $.
\end{enumerate}
\end{corollary}

\section{Constructions} \label{sec constructions}

In this section, we provide several methods of constructing codes for the multi-cover metric, focusing on linear MMCD codes and dually MMCD codes.

\subsection{Codes in the sum-rank metric} \label{subsec sum-rank}

Similarly to the case of the rank metric and the cover metric \cite{gabidulin-crisscross, roth}, we show in this subsection that sum-rank metric codes may be used as multi-cover metric codes. The sum-rank metric was formally defined in \cite[Sec. III-D]{multishot}, but it was implicitly used earlier in \cite[Sec. III]{spacetime-kumar}. 

\begin{definition} [\textbf{Sum-rank metric \cite{spacetime-kumar, multishot}}]
We define the sum-rank weight of $ C = (C_1, C_2, \ldots, C_\ell) \in \prod_{i=1}^\ell \mathbb{F}_q^{m_i \times n_i} $ as
$$ {\rm wt}_{SR}(C) = \sum_{i=1}^\ell {\rm Rk}(C_i). $$
The sum-rank metric is then defined as $ {\rm d}_{SR} : (\prod_{i=1}^\ell \mathbb{F}_q^{m_i \times n_i})^2 \longrightarrow \mathbb{N} $, where $ {\rm d}_{SR}(C,D) = {\rm wt}_{SR}(C-D) $, for $ C,D \in \prod_{i=1}^\ell \mathbb{F}_q^{m_i \times n_i} $.
\end{definition}

The bound in Theorem \ref{th singleton} is also valid for the sum-rank metric \cite[Th. III.2]{alberto-fundamental}, and a code attaining it is called maximum sum-rank distance (MSRD).

We have the following connections between both metrics. They constitute a trivial extension to $ \ell \geq 1 $ of the corresponding results for the case $ \ell = 1 $, which were first observed in \cite{gabidulin-crisscross, roth}. We note that Item 4 follows from combining Item 3 with the fact that the dual of a linear MSRD code is again MSRD code under the given conditions \cite[Th. VI.1]{alberto-fundamental}.

\begin{proposition} \label{prop multi-cover from sum-rank}
Fix $ C \in \prod_{i=1}^\ell \mathbb{F}_q^{m_i \times n_i} $ and $ \mathcal{C} \subseteq \prod_{i=1}^\ell \mathbb{F}_q^{m_i \times n_i} $. The following hold:
\begin{enumerate}
\item
$ {\rm wt}_{SR}(C) \leq {\rm wt}_{MC} (C) $.
\item
$ {\rm d}_{SR} (\mathcal{C}) \leq {\rm d}_{MC}(\mathcal{C}) $.
\item
If $ \mathcal{C} $ is an MSRD code, then it is also MMCD.
\item
If $ \mathcal{C} $ is a linear MSRD code and $ m_1 = m_2 = \ldots = m_\ell $, then $ \mathcal{C} $ is a dually MMCD code.
\end{enumerate}
\end{proposition}

We now show that sum-rank error and erasure correcting algorithms may be used to correct multilayer crisscross errors and erasures. There exist two equivalent formulations of error and erasure correction in the sum-rank metric. We consider both column and row erasures, which was first considered in \cite{sven-generic}. Item 2 is included since it was the formulation used in \cite{hormann, sven-generic}, but the connection with the multi-cover metric is easier using Item 3. The equivalence between the two is proven as in \cite[Prop. 17]{similarities}. 

\begin{proposition} \label{prop sum-rank errors and erasures}
Let $ \mathcal{C} \subseteq \prod_{i=1}^\ell \mathbb{F}_q^{m_i \times n_i} $ be a code, and let $ \rho_R $, $ \rho_C $ and $ t $ be non-negative integers. The following are equivalent:
\begin{enumerate}
\item
$ 2t + \rho_R + \rho_C < {\rm d}_{SR}(\mathcal{C}) $.
\item
For $ \mathcal{V}_R = (\mathcal{V}_{R,i})_{i=1}^\ell $ and $ \mathcal{V}_C = (\mathcal{V}_{C,i})_{i=1}^\ell $, where $ \mathcal{V}_{R,i} \subseteq \mathbb{F}_q^{n_i} $ and $ \mathcal{V}_{C,i} \subseteq \mathbb{F}_q^{m_i} $ are vector subspaces (known to the receiver), for $ i = 1,2, \ldots, \ell $, with $ \sum_{i=1}^\ell \dim(\mathcal{V}_{R,i}) \leq \rho_R $ and $ \sum_{i=1}^\ell \dim(\mathcal{V}_{C,i}) \leq \rho_C $, there exists a decoder $ D_{\mathcal{V}_R, \mathcal{V}_C,t} : \prod_{i=1}^{\ell} \mathbb{F}_q^{m_i \times n_i} \longrightarrow \mathcal{C} $ such that
$$ D_{\mathcal{V}_R, \mathcal{V}_C,t}(C + E + E_R + E_C) = C, $$
for all $ C \in \mathcal{C} $ and all $ E,E_R,E_C \in \prod_{i=1}^{\ell} \mathbb{F}_q^{m_i \times n_i} $ such that $ {\rm wt}_{SR}(E) \leq t $, $ {\rm wt}_{SR}(E_R) \leq \rho_R $ and $ {\rm wt}_{SR}(E_C) \leq \rho_C $, and the row space of the $ i $th component of $ E_R $ and column space of the $ i $th component of $ E_C $ are $ \mathcal{V}_{R,i} $ and $ \mathcal{V}_{C,i} $, respectively.
\item
For partitions $ \rho_R = \sum_{i=1}^\ell \rho_{R,i} $ and $ \rho_C = \sum_{i=1}^\ell \rho_{C,i} $ into non-negative integers, and for block-diagonal matrices $ A = {\rm diag}(A_1, A_2, \ldots, A_\ell) $ and $ B = {\rm diag}(B_1, B_2, \ldots, B_\ell) $, for full-rank matrices $ A_i \in \mathbb{F}_q^{(m_i - \rho_{R,i}) \times m_i} $ and $ B_i \in \mathbb{F}_q^{n_i \times (n_i - \rho_{C,i})} $ (known to the receiver), there exists a decoder $ D_{A,B,t} : \prod_{i=1}^{\ell} \mathbb{F}_q^{(m_i - \rho_{R,i}) \times (n_i - \rho_{C,i})} \longrightarrow \mathcal{C} $ such that
$$ D_{A,B,t}(ACB + E) = C, $$
for all $ C \in \mathcal{C} $ and all $ E \in \prod_{i=1}^{\ell} \mathbb{F}_q^{(m_i - \rho_{R,i}) \times (n_i - \rho_{C,i})} $ such that $ {\rm wt}_{SR}(E) \leq t $.
\end{enumerate}
\end{proposition}

In particular, we deduce that codes able to correct sum-rank errors and erasures (in rows and columns) may without any change also correct multilayer crisscross errors and erasures as in Definition \ref{def decoders for crisscross}. The proof of the following corollary is straightforward by Item 3 in Proposition \ref{prop sum-rank errors and erasures} and Proposition \ref{prop multi-cover from sum-rank}.

\begin{corollary} \label{cor decoding from SR to MC}
Let $ \mathcal{C} \subseteq \prod_{i=1}^\ell \mathbb{F}_q^{m_i \times n_i} $ be a code, and let $ \rho_R $, $ \rho_C $ and $ t $ be non-negative integers with $ 2t + \rho_R + \rho_C < {\rm d}_{SR}(\mathcal{C}) $. Let $ X = (X_i, Y_i)_{i=1}^\ell \in {\rm MC}(\mathbf{m},\mathbf{n}) $ be such that $ \sum_{i=1}^\ell |X_i| = \rho_R $ and $ \sum_{i=1}^\ell |Y_i| = \rho_C $. Let $ A_i \in \mathbb{F}_q^{(m_i - |X_i|) \times m_i} $ and $ B_i \in \mathbb{F}_q^{n_i \times (n_i - |Y_i|)} $ be the projection matrices onto the coordinates outside of $ X_i $ and $ Y_i $, respectively, for $ i = 1,2, \ldots, \ell $. A sum-rank decoder $ D_{A,B,t} : \prod_{i=1}^{\ell} \mathbb{F}_q^{(m_i - |X_i|) \times (n_i - |Y_i|)} \longrightarrow \mathcal{C} $ as in Item 3 in Proposition \ref{prop sum-rank errors and erasures} may also be used without a change as a multi-cover decoder $ D_{X,t} = D_{A,B,t} : \prod_{i=1}^{\ell} \mathbb{F}_q^{(m_i - |X_i|) \times (n_i - |Y_i|)} \longrightarrow \mathcal{C} $ as in Definition \ref{def decoders for crisscross}.
\end{corollary}

\subsection{A nested construction} \label{subsec nested}

In this subsection, we provide a general method to construct codes for the multi-cover metric from other multi-cover metric codes. The idea is to adequately arrange the components of a Cartesian product. Throughout this subsection, we fix positive integers $ m = ur $, $ n = us $ and $ t = u \ell $, and we will assume that $ n \leq m $ (thus $ s \leq r $).

\begin{construction} \label{const nested}
Let $ \mathcal{C} \subseteq (\mathbb{F}_q^{r \times s})^t $ be a code. We define another code $ \varphi( \mathcal{C} ) \subseteq (\mathbb{F}_q^{m \times n})^\ell $ as the image of the linear map $ \varphi : ((\mathbb{F}_q^{r \times s})^t )^u \longrightarrow (\mathbb{F}_q^{m \times n})^\ell $, where $ \varphi \left( C^1, C^2, \ldots, C^u \right) = $
$$
\resizebox{\textwidth}{!}{ $
\left(
\begin{array}{cccc|cccc|c|cccc}
C_1^1 & C_1^2 & \ldots & C_1^u & C_{u+1}^1 & C_{u+1}^2 & \ldots & C_{u+1}^u & \ldots & C_{(\ell-1)u+1}^1 & C_{(\ell-1)u+1}^2 & \ldots & C_{(\ell-1)u+1}^u \\
C_2^u & C_2^1 & \ldots & C_2^{u-1} & C_{u+2}^u & C_{u+2}^1 & \ldots & C_{u+2}^{u-1} & \ldots & C_{(\ell-1)u+2}^u & C_{(\ell-1)u+2}^1 & \ldots & C_{(\ell-1)u+2}^{u-1} \\
\vdots & \vdots & \ddots & \vdots & \vdots & \vdots & \ddots & \vdots & \ddots & \vdots & \vdots & \ddots & \vdots \\
C_u^2 & C_u^3 & \ldots & C_u^1 & C_{2u}^2 & C_{2u}^3 & \ldots & C_{2u}^1 & \ldots & C_t^2 & C_t^3 & \ldots & C_t^1 \\
\end{array}
\right),
$}
$$
for $ C^i = (C_1^i, C_2^i, \ldots, C_t^i) \in (\mathbb{F}_q^{r \times s})^t $, for $ i = 1,2, \ldots, u $.
\end{construction}

We now relate the multi-cover metric parameters of $ \mathcal{C} $ and $ \varphi(\mathcal{C}) $.

\begin{theorem} \label{th nested properties}
Let $ \mathcal{C} \subseteq (\mathbb{F}_q^{r \times s})^t $. The following hold:
\begin{enumerate}
\item
$ {\rm d}_{MC}(\varphi(\mathcal{C})) = {\rm d}_{MC}(\mathcal{C}) $ and $ |\varphi(\mathcal{C})| = |\mathcal{C}|^u $.
\item
$ \varphi(\mathcal{C}) $ is MMCD if, and only if, so is $ \mathcal{C} $.
\item
$ \varphi(\mathcal{C}) $ is linear if, and only if, so is $ \mathcal{C} $, and in that case, $ \dim(\varphi(\mathcal{C})) = u \dim(\mathcal{C}) $ and $ \varphi(\mathcal{C})^\perp = \varphi(\mathcal{C}^\perp) $.
\item
(If $ \mathcal{C} $ is linear) $ \varphi(\mathcal{C}) $ is a dually MMCD code if, and only if, so is $ \mathcal{C} $.
\end{enumerate}
\end{theorem}
\begin{proof}
The equality $ |\varphi(\mathcal{C})| = |\mathcal{C}|^u $ and Item 3 are straightforward, since $ \varphi(\mathcal{C}) \cong \mathcal{C}^u $ as vector spaces and the inner product is preserved through a natural isomorphism. We next prove that $ {\rm d}_{MC}(\varphi(\mathcal{C})) = {\rm d}_{MC}(\mathcal{C}) $.

Let $ C = \varphi \left( C^1, C^2, \ldots, C^u \right) $, where $ C^i \in \mathcal{C} $, for $ i = 1,2, \ldots, u $. Let $ (X_i,Y_i)_{i=1}^\ell \in {\rm MC}(C) $. Partition $ X_i = \bigcup_{j=1}^u X_{i,j} $ and $ Y_i = \bigcup_{j=1}^u Y_{i,j} $, where $ X_{i,j} \subseteq [(j-1)r+1, jr] $ and $ Y_{i,j} \subseteq [(j-1)s+1, js] $, for $ j = 1,2, \ldots, u $ and $ i = 1,2, \ldots, \ell $. If we identify $ X_{i,j} $ and $ Y_{i,j} $ in the obvious way with subsets of $ [r] $ and $ [s] $, respectively, then it is easy to see that $ (X_{i,j}, Y_{i,j})_{i=1,j=1}^{\ell, u} \in {\rm MC} (C^k) $, for $ k = 1,2, \ldots, u $. Assume that $ {\rm wt}_{MC}(C) = {\rm d}_{MC}(\varphi(\mathcal{C})) $ and $ (X_i,Y_i)_{i=1}^\ell \in {\rm MC}(C) $ is of minimum size. Therefore, if $ C^1 \neq 0 $,
$$ {\rm d}_{MC} (\varphi(\mathcal{C})) = {\rm wt}_{MC}(\mathcal{C}) = \sum_{i=1}^\ell (|X_i|+|Y_i|) $$
$$ = \sum_{i=1}^\ell \sum_{j=1}^u (|X_{i,j}| + |Y_{i,j}|) \geq {\rm wt}_{MC}(C^1) \geq {\rm d}_{MC}(\mathcal{C}) . $$

Conversely, let $ D \in \mathcal{C} $, and let $ (X^\prime_k, Y^\prime_k)_{k=1}^t \in {\rm MC}(D) $. Define $ X_{i,j} = (j-1)r + X^\prime_{(i-1) \ell +j} \subseteq [(j-1)r+1, jr] $ and $ Y_{i,j} = (j-1)s + Y^\prime_{(i-1) \ell +j} \subseteq [(j-1)s+1, js] $, for $ j = 1,2, \ldots, u $ and $ i = 1,2, \ldots, \ell $. Finally define $ X_i = \bigcup_{j=1}^u X_{i,j} $ and $ Y_i = \bigcup_{j=1}^u Y_{i,j} $, for $ i = 1,2, \ldots, \ell $. Then $ (X_i,Y_i)_{i=1}^\ell \in {\rm MC} (\varphi(D, 0, \ldots, 0 )) $, where $ \varphi(D, 0, \ldots, 0 ) \in \varphi(\mathcal{C}) $. Assume that $ {\rm wt}_{MC}(D) = {\rm d}_{MC}(\mathcal{C}) $ and $ (X^\prime_k, Y^\prime_k)_{k=1}^t \in {\rm MC}(D) $ is of minimum size. Hence, we conclude that 
$$ {\rm d}_{MC}(\mathcal{C}) = {\rm wt}_{MC}(D) = \sum_{k=1}^t (|X^\prime_k|+|Y^\prime_k|) $$
$$ = \sum_{i=1}^\ell (|X_i|+|Y_i|) \geq {\rm wt}_{MC} (\varphi(D, 0, \ldots, 0 )) \geq {\rm d}_{MC}(\varphi(\mathcal{C})). $$

Thus Item 1 is proven. We now prove Item 2. Assume that $ \mathcal{C} $ is MMCD, that is, $ |\mathcal{C}| = q^{r(ts - d + 1)} $, where $ d = {\rm d}_{MC}(\mathcal{C}) = {\rm d}_{MC}(\varphi(\mathcal{C})) $. Since $ m = ur $, $ n =us $ and $ t = u \ell $, we have that
$$ |\varphi(\mathcal{C})| = |\mathcal{C}|^u = q^{ur(ts - d+1)} = q^{m(\ell n - d + 1)}, $$
hence $ \varphi(\mathcal{C}) $ is also MMCD. For the reversed implication, we may show in the same way that if $ |\mathcal{C}| < q^{r(ts - d + 1)} $, then $ |\varphi(\mathcal{C})| < q^{m(\ell n - d + 1)} $. 

Finally, Item 4 follows from Items 2 and 3.
\end{proof}

In this way, we may construct codes in $ (\mathbb{F}_q^{m \times n})^\ell $ for the multi-cover metric from codes in $ (\mathbb{F}_q^{r \times s})^t $ for the refined multi-cover metric. Note that, if we set $ r = s = 1 $, then we may construct multi-cover metric codes in $ (\mathbb{F}_q^{n \times n})^\ell $ from Hamming-metric codes in $ \mathbb{F}_q^t $.

In addition, a decoder for $ \mathcal{C} $ may be directly used to obtain a decoder for $ \varphi(\mathcal{C}) $, as we now show. The proof is straightforward (but cumbersome) using the partitions from the proof of Theorem \ref{th nested properties}, and is left to the reader for brevity.

\begin{proposition} \label{prop decoding nested}
Let $ X = (X_i,Y_i)_{i=1}^\ell \in {\rm MC}(\mathbf{m},\mathbf{n}) $, partition $ X_i = \bigcup_{j=1}^u X_{i,j} $ and $ Y_i = \bigcup_{j=1}^u Y_{i,j} $, where $ X_{i,j} \subseteq [(j-1)r+1, jr] $ and $ Y_{i,j} \subseteq [(j-1)s+1, js] $, for $ j = 1,2, \ldots, u $ and $ i = 1,2, \ldots, \ell $. Next identify $ X_{i,j} $ and $ Y_{i,j} $ in the obvious way with subsets of $ [r] $ and $ [s] $, respectively, and define $ X^\prime = (X_{i,j}, Y_{i,j})_{i=1,j=1}^{\ell, s} $. Set $ \rho = |X| = |X^\prime| $. 

Let $ \mathcal{C} \subseteq (\mathbb{F}_q^{r \times s})^t $. Assume that $ 2t + \rho < d = {\rm d}_{MC}(\varphi(\mathcal{C})) = {\rm d}_{MC}(\mathcal{C}) $. Let $ D_{\mathcal{C},X^\prime} : (\mathbb{F}_q^{r \times s})^t \longrightarrow \mathcal{C} $ be a $ t $-error and $ \rho $-erasure-correcting decoder for $ \mathcal{C} $ and $ X^\prime $. Define $ D_{\varphi(\mathcal{C}),X} : (\mathbb{F}_q^{m \times n})^\ell \longrightarrow \varphi(\mathcal{C}) $ as the decoder that decodes $ \varphi(C^1, C^2, \ldots, C^u) $ by using $ D_{\mathcal{C},X^\prime} $ component-wise on the codewords $ C^k \in \mathcal{C} $, for $ k = 1,2, \ldots, u $. Then $ D_{\varphi(\mathcal{C}), X} $ is a $ t $-error and $ \rho $-erasure-correcting decoder for $ \varphi(\mathcal{C}) $ and $ X $.
\end{proposition}

\subsection{Some explicit codes} \label{subsec explicit codes}

In this subsection, we put together the two methods for constructing multi-cover metric codes from Subsections \ref{subsec sum-rank} and \ref{subsec nested} in order to give explicit codes in the multi-cover metric. We consider positive integers $ n = us $ and $ t = u \ell $, as in the previous subsection (we will consider $ m = n $ and $ r = s $, i.e., square matrices).

We start by providing explicit families of dually MMCD codes for general parameters. We consider linearized Reed-Solomon codes \cite[Def. 31]{linearizedRS} as component codes in Construction \ref{const nested}, for several choices of the integer $ t $. Consider an ordered basis $ \boldsymbol\alpha = ( \alpha_1, \alpha_2, \ldots, \alpha_s ) \in \mathbb{F}_{q^s} $ of $ \mathbb{F}_{q^s} $ over $ \mathbb{F}_q $, and define $ M_{\boldsymbol\alpha} : \mathbb{F}_{q^s}^s \longrightarrow \mathbb{F}_q^{s \times s} $ by 
\begin{equation}
M_{\boldsymbol\alpha} \left( \mathbf{c} \right) = \left( \begin{array}{cccc}
c_{1,1} & c_{1,2} & \ldots & c_{1,s} \\
c_{2,1} & c_{2,2} & \ldots & c_{2,s} \\
\vdots & \vdots & \ddots & \vdots \\
c_{s,1} & c_{s,2} & \ldots & c_{s,s} \\
\end{array} \right) \in \mathbb{F}_q^{s \times s},
\label{eq def matrix representation map}
\end{equation}
where $ \mathbf{c} = \sum_{i=1}^s \alpha_i (c_{i,1}, c_{i,2}, \ldots, c_{i,s}) \in \mathbb{F}_{q^s}^s $, for $ i = 1,2, \ldots, s $. We extend this map to $ t $-tuples of matrices as 
\begin{equation}
M_{t ,\boldsymbol\alpha} \left( \mathbf{c}^{(1)}, \mathbf{c}^{(2)}, \ldots , \mathbf{c}^{(t)} \right) = \left( M_{\boldsymbol\alpha} \left( \mathbf{c}^{(1)} \right), M_{\boldsymbol\alpha} \left( \mathbf{c}^{(2)} \right), \ldots, M_{\boldsymbol\alpha} \left( \mathbf{c}^{(t)} \right) \right) \in (\mathbb{F}_q^{s \times s})^t,
\label{eq def matrix repr map tuples}
\end{equation}
where $ \mathbf{c}^{(i)} \in \mathbb{F}_{q^s}^s $, for $ i = 1,2, \ldots, t $. We will consider $ \mathcal{C} = M_{t ,\boldsymbol\alpha}(\mathcal{C}_{ts,k}) \subseteq (\mathbb{F}_q^{s \times s})^t $, where $ \mathcal{C}_{ts,k} \subseteq \mathbb{F}_{q^s}^{ts} $ is the linearized Reed--Solomon code with generator matrix $ G_{ts,k} = (G_1 | G_2 | \ldots | G_t) \in \mathbb{F}_{q^s}^{k \times (ts)} $, where 
\begin{equation*}
G_i = \left( \begin{array}{cccc}
\beta_1 & \beta_2 & \ldots & \beta_s \\
\beta_1^q \gamma^{i-1} & \beta_2^q \gamma^{i-1} & \ldots & \beta_s^q \gamma^{i-1} \\
\beta_1^{q^2} \gamma^{(i-1)\frac{q^2-1}{q-1}} & \beta_2^{q^2} \gamma^{(i-1)\frac{q^2-1}{q-1}} & \ldots & \beta_s^{q^2} \gamma^{(i-1)\frac{q^2-1}{q-1}} \\
\vdots & \vdots & \ddots & \vdots \\
\beta_1^{q^{k-1}} \gamma^{(i-1)\frac{q^{k-1}-1}{q-1}} & \beta_2^{q^{k-1}} \gamma^{(i-1)\frac{q^{k-1}-1}{q-1}} & \ldots & \beta_s^{q^{k-1}} \gamma^{(i-1)\frac{q^{k-1}-1}{q-1}}
\end{array} \right) \in \mathbb{F}_{q^s}^{k \times s},
\end{equation*}
for $ i = 1,2, \ldots, t $, where $ \beta_1, \beta_2, \ldots, \beta_s \in \mathbb{F}_{q^s} $ form a basis of $ \mathbb{F}_{q^s} $ over $ \mathbb{F}_q $ and $ \gamma \in \mathbb{F}_{q^s}^* $ is a primitive element, i.e., $ \mathbb{F}_{q^s}^* = \{ 1, \gamma, \gamma^2, \ldots, \gamma^{q^s-2} \} $, which always exists by \cite[Th. 2.8]{lidl}. 

If $ q > t $, then the linear code $ \mathcal{C} \subseteq (\mathbb{F}_q^{s \times s})^t $ above is MSRD \cite[Th. 4]{linearizedRS}. Hence we deduce the following result.

\begin{theorem} \label{th MMCD code from linearized RS}
Assume that $ q > t $. For any $ k = 1,2, \ldots, st = \ell n $, let $ \mathcal{C} \subseteq (\mathbb{F}_q^{s \times s})^t $ be as above. Then the code $ \varphi(\mathcal{C}) \subseteq (\mathbb{F}_q^{n \times n})^\ell $ obtained from $ \mathcal{C} $ as in Construction \ref{const nested} is a dually MMCD code of dimension $ \dim(\varphi(\mathcal{C})) = nk $. Furthermore, it can correct any $ t $ errors and $ \rho $ erasures for the multi-cover metric, as in Definition \ref{def decoders for crisscross}, where $ 2t + \rho \leq \ell n - k $, with a complexity of $ \mathcal{O}(t \ell n^2) $ sums and products over the finite field of size $ q^{\ell n / t} = \mathcal{O}(t)^{\ell n / t} $.
\end{theorem}
\begin{proof}
Since $ \mathcal{C} $ is MSRD \cite[Th. 4]{linearizedRS}, we have that $ \varphi(\mathcal{C}) $ is dually MMCD by Proposition \ref{prop multi-cover from sum-rank} and Theorem \ref{th nested properties}. 

Second, $ \mathcal{C} $ can correct $ t $ errors and $ \rho $ erasures (in rows or columns) for the sum-rank metric, as in Item 2 in Proposition \ref{prop sum-rank errors and erasures}, with a complexity of $ \mathcal{O}(s^2t^2) $ sums and products over the finite field $ \mathbb{F}_{q^s} $, by the algorithm in \cite{hormann}. By Corollary \ref{cor decoding from SR to MC} and Proposition \ref{prop decoding nested}, we may decode $ \varphi(\mathcal{C}) $ by using such a decoder $ u $ times. Hence the number of operations over $ \mathbb{F}_{q^s} $ is $ \mathcal{O}(rs^2t^2) = \mathcal{O}(t \ell n^2) $. Since we may choose $ q = \mathcal{O}(t) $, then $ q^s = q^{\ell n / t} = \mathcal{O}(t)^{\ell n / t} $, and we are done.
\end{proof}

Assume that a multiplication in $ \mathbb{F}_{2^b} $ costs $ \mathcal{O}(b^2) $ operations in $ \mathbb{F}_2 $. Then if $ q $ is even, the dually MMCD code in Theorem \ref{th MMCD code from linearized RS} can be decoded with 
$$ \mathcal{O} \left( t^{-1}\log_2(t+1)^2 \ell^3 n^4 \right) $$
operations over $ \mathbb{F}_2 $. This complexity is smaller for larger values of $ t $. However, the alphabet size for the multi-cover metric needs to satisfy $ q > t $. Thus we arrive at an \textit{alphabet-complexity trade-off}: Codes for larger $ t $ are faster to decode but require larger alphabets (i.e., can be applied to a smaller set of alphabets), whereas codes for smaller $ t $ are less fast but can be used for a wider range of alphabets. 

\begin{remark}
Note that, for $ t = 1 $, Theorem \ref{th MMCD code from linearized RS} corresponds to using Gabidulin codes as MMCD codes via Construction \ref{const nested}, since in this case linearized Reed-Solomon codes recover Gabidulin codes \cite{linearizedRS}. This is the only choice of $ t $ for the alphabet $ \mathbb{F}_2 $ but is also the least computationally efficient choice. Further setting $ \ell = 1 $, we obtain the classical rank-metric construction for the cover metric \cite{gabidulin-crisscross, roth}. On the other end, if we set $ t = \ell n $ (i.e., $ s = 1 $), then Theorem \ref{th MMCD code from linearized RS} corresponds to using classical Reed-Solomon codes as MMCD codes via Construction \ref{const nested}, since in this case linearized Reed-Solomon codes recover classical Reed-Solomon codes \cite{linearizedRS}. Further setting $ \ell = 1 $, we also recover the construction of cover-metric codes based on Reed-Solomon codes from \cite{roth}.
\end{remark}

\begin{remark}
There exist other families of MSRD codes for further parameter regimes \cite{generalMSRD}. In particular, those from \cite{generalMSRD} attain smaller field sizes than linearized Reed-Solomon codes for several parameter regimes, and have therefore a better potential of having a faster decoder. However such MSRD codes do not work for square matrices $ m = n $. All of these MSRD codes may be turned into MMCD codes as in Theorem \ref{th MMCD code from linearized RS}. We leave the details to the reader for brevity.
\end{remark}

Finally, we describe a family of codes in the multi-cover metric obtained from sum-rank BCH codes \cite{SR-BCH}, which contain the best codes in the sum-rank metric for the binary field $ \mathbb{F}_2 $ and $ 2 \times 2 $ matrices (see \cite[App.]{SR-BCH}). We will use the simplified bound on their parameters from \cite[Th. 9]{SR-BCH}, and we will not explicitly describe the codes for brevity.

\begin{proposition} \label{prop SR-BCH codes}
Let $ n = us $, $ t = u \ell $ and $ 0 \leq \delta \leq \ell n $ be positive integers. Assume that $ t $ and $ s $ are coprime, $ t $ and $ q $ are coprime, and $ t $ divides $ q - 1 $. Let $ q_0 $ and $ r $ be positive integers such that $ q = q_0^r $. Let $ x^t - 1 = m_1(x) m_2(x) \cdots m_v(x) $ be the irreducible decomposition of $ x^t-1 $ in $ \mathbb{F}_{q_0^s} $. Then there exists a linear code $ \mathcal{C} \subseteq (\mathbb{F}_{q_0}^{n \times n})^\ell $ such that $ {\rm d}_{MC}(\mathcal{C}) \geq \delta $ and
\begin{equation}
\dim(\mathcal{C}) \geq \ell n^2 - n \sum_{i=1}^v \min \left\lbrace s d_i , rk_i \right\rbrace ,
\label{eq lower bound dim SR BCH}
\end{equation}
where $ d_i = \deg(m_i(x)) $ and $ k_i = | \{ j \in \mathbb{N} \mid 0 \leq j \leq \delta - 2, m_i(a^{b+j}) = 0 \} | $, for $ i = 1,2, \ldots, v $, and where $ a \in \mathbb{F}_q $ is a primitive $ t $-th root of unity, which always exists \cite[Sec. 2.4]{lidl} \cite[Page 105]{pless}.
\end{proposition}
\begin{proof}
With the assumptions in the proposition, it holds by \cite[Th. 9]{SR-BCH} that there exists an $ \mathbb{F}_{q_0^s} $-linear code $ \mathcal{D} \subseteq \mathbb{F}_{q_0^s}^{ts} $ with ${\rm d}_{SR}(\mathcal{D}) \geq \delta $ and 
$$ \dim_{\mathbb{F}_{q_0^s}}(\mathcal{D}) \geq ts - \sum_{i=1}^v \min \left\lbrace s d_i , rk_i \right\rbrace . $$
Now, define $ \mathcal{C} = \varphi(M_{t,\boldsymbol\alpha}(\mathcal{D})) \subseteq (\mathbb{F}_{q_0}^{n \times n})^\ell $, with $ \varphi $ as in Construction \ref{const nested} and $ M_{t,\boldsymbol\alpha} $ as in (\ref{eq def matrix repr map tuples}) for an ordered basis $ \boldsymbol\alpha \in \mathbb{F}_{q_0^s}^s $ of $ \mathbb{F}_{q_0^s} $ over $ \mathbb{F}_{q_0} $. By Proposition \ref{prop multi-cover from sum-rank} and Theorem \ref{th nested properties}, we have that 
$$ {\rm d}_{MC}(\mathcal{C}) = {\rm d}_{MC}(M_{t,\boldsymbol\alpha}(\mathcal{D})) \geq {\rm d}_{SR}(M_{t,\boldsymbol\alpha}(\mathcal{D})) = {\rm d}_{SR}(\mathcal{D}) \geq \delta . $$
Finally, by Theorem \ref{th nested properties}, we have that 
$$ \dim(\mathcal{C}) = u \dim(M_{t,\boldsymbol\alpha}(\mathcal{D})) = us \dim_{\mathbb{F}_{q_0^s}}(\mathcal{D}) $$  
$$ \geq us \left( ts - \sum_{i=1}^v \min \left\lbrace s d_i , rk_i \right\rbrace \right) = \ell n^2 - n \sum_{i=1}^v \min \left\lbrace s d_i , rk_i \right\rbrace , $$
and the result follows.
\end{proof}

The codes in Proposition \ref{prop SR-BCH codes} have an advantage over general subfield subcodes of linear MMCD codes, due to the following. One may consider Delsarte's lower bound on the dimension of a subfield subcode \cite{delsarte}. More concretely, let $ \mathcal{D} \subseteq (\mathbb{F}_q^{n \times n})^\ell $ be a linear MMCD code of dimension $ n (\delta - 1) $, $ q = q_0^r $ and $ \mathcal{C} = \mathcal{D}^\perp \cap (\mathbb{F}_{q_0}^{n \times n})^\ell $, for positive integers $ \delta $ and $ r $. Then $ {\rm d}_{MD}(\mathcal{C}) \geq \delta $ and Delsarte's bound \cite{delsarte} states that
\begin{equation}
\dim(\mathcal{C}) \geq \ell n^2 - nr (\delta - 1).
\label{eq lower bound delsarte}
\end{equation}
However, as shown in \cite[Sec. VII-C]{SR-BCH}, the lower bound (\ref{eq lower bound dim SR BCH}) is tighter than (\ref{eq lower bound delsarte}) in all cases. See also \cite[App.]{SR-BCH} for numerical tables.

Finally, as shown in \cite[Sec. VII-D]{SR-BCH}, a sum-rank BCH code (used in Proposition \ref{prop SR-BCH codes}) may be decoded for the sum-rank metric by decoding the corresponding linearized Reed--Solomon code. Since these latter codes may correct sum-rank errors and erasures (by rows and columns) by \cite{hormann}, then the same holds for the codes in Proposition \ref{prop SR-BCH codes}, by Corollary \ref{cor decoding from SR to MC} and Proposition \ref{prop decoding nested}.

\section{Conclusion and open problems} \label{sec conclusion}

In this work, we considered multilayer crisscross error and erasure correction, which affect entire rows and columns of matrices accross a list of matrices. We introduced the multi-cover metric to measure such errors and erasures. We derived bounds for this metric, including a Singleton-like bound, and introduced codes attaining it, called maximum multi-cover distance (MMCD) codes. We study their duality, puncturing and shortening properties, and then we explored several constructions of dually MMCD codes, together with their decoding. As open problems, it would be interesting to generalize to the multi-cover metric the following works in the classical cover metric:
\begin{enumerate}
\item
In this work, we considered error-free worst-case deterministic decoding. Probabilistic decoding in the cover metric was considered in \cite{roth-probabilistic}, where a low-redundancy probabilistic decoder allowing a small probability of error was presented.
\item
List decoding for the cover metric was first studied in \cite{wachter-crisscross}, where a Johnson bound was derived and an efficient list-decoding algorithm (up to the bound) was presented.
\item
Crisscross insertions and deletions were studied recently in \cite{bitar-crisscross, welter-crisscross}, where several code constructions are given.
\item
Codes with local crisscross erasure correcting properties were studied in \cite{kadhe-crisscross, liu-crisscross}. The work \cite{kadhe-crisscross} focuses on codes with locality for the rank metric, allowing to locally correct full columns, whereas a more general case is considered in \cite{liu-crisscross}.
\end{enumerate}

\section*{Acknowledgement}

The author gratefully acknowledges the support from a Mar{\'i}a Zambrano contract by the University of Valladolid, Spain (Contract no. E-47-2022-0001486).


 
\bibliographystyle{plain}

\appendix

\section*{Appendix: Proofs for Subsection \ref{subsec size ball}} \label{app size ball}

\begin{proof}[Proof of Lemma \ref{lemma size full sphere}]
The equality $ S_1^{m,1} = q^m - 1 $ is trivial. Now, let $ \mathcal{C} \subseteq \mathbb{F}_q^{m \times n} $ be the set of matrices with no zero columns, and let $ \mathcal{R} \subseteq \mathbb{F}_q^{m \times n} $ be the set of matrices with no zero rows. Then $ S_n^{m,n} = | \mathcal{C} \cup \mathcal{R} | $. Hence, we only need to prove that
$$ |(\mathcal{C} \cup \mathcal{R})^c| = | \mathcal{C}^c \cap \mathcal{R}^c | = \sum_{i=1}^m \sum_{j=1}^n \binom{m}{i} \binom{n}{j} S_{\min \{ m-i, n-j \}}^{m-i, n-j}, $$
where $ \mathcal{A}^c $ denotes the complement of a subset $ \mathcal{A} \subseteq \mathbb{F}_q^{m \times n} $ in $ \mathbb{F}_q^{m \times n} $. Note that $ |\mathcal{C}^c \cap \mathcal{R}^c| $ is the number of matrices in $ \mathbb{F}_q^{m \times n} $ with at least one zero row and at least one zero column. For $ i = 1,2, \ldots, m $ and $ j = 1,2, \ldots, n $, the number of matrices in $ \mathbb{F}_q^{m \times n} $ with exactly $ i $ zero rows and exactly $ j $ zero columns is
$$ \binom{m}{i} \binom{n}{j} S_{\min \{ m-i, n-j \}}^{m-i, n-j}. $$
Hence the recursive formula follows.

Finally, the inequality $ S_n^{m,n} \leq q^{mn} $ is trivial, and the inequality $ (q-1)^{mn} \leq S_n^{m,n} $ follows from the fact that any matrix in $ \mathbb{F}_q^{m \times n} $ with no zero entries has cover weight $ n $.
\end{proof}

\begin{proof} [Proof of Theorem \ref{th size ball}]
Let $ (X,Y) \in {\rm MC}(m,n) $ with $ |X| + |Y| = r $. Set $ s = |X| $, thus $ r-s = |Y| $. The number of matrices having $ (X,Y) $ as a minimal cover (thus of cover weight $ r $) is exactly $ (q^{m-s}-1)^{r-s}(q^{n-r+s}-1)^s q^{s(r-s)} $. To see this, consider Fig. \ref{fig ball size upper bound}.

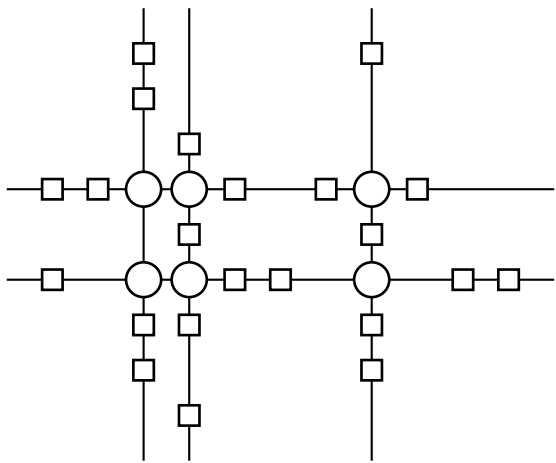
\begin{figure} [!h]
\begin{center}
\begin{tabular}{c@{\extracolsep{1cm}}c}
	\begin{tikzpicture}[line width=1pt, scale=0.6]
		\tikzstyle{every node}=[inner sep=0pt, minimum width=4.5pt]
		

		\draw[thick] (3,0) -- (3,10);
		\draw[thick] (4,0) -- (4,10);
		\draw[thick] (8,0) -- (8,10);
		
		\draw[thick] (0,4) -- (12,4);
		\draw[thick] (0,6) -- (12,6);
		
		\draw (3,2) node () [draw, rectangle, fill=white, minimum size=0.7em] {};
		\draw (3,3) node () [draw, rectangle, fill=white, minimum size=0.7em] {};
		\draw (3,8) node () [draw, rectangle, fill=white, minimum size=0.7em] {};
		\draw (3,9) node () [draw, rectangle, fill=white, minimum size=0.7em] {};
		
		\draw (4,1) node () [draw, rectangle, fill=white, minimum size=0.7em] {};
		\draw (4,3) node () [draw, rectangle, fill=white, minimum size=0.7em] {};
		\draw (4,5) node () [draw, rectangle, fill=white, minimum size=0.7em] {};
		\draw (4,7) node () [draw, rectangle, fill=white, minimum size=0.7em] {};
		
		\draw (8,2) node () [draw, rectangle, fill=white, minimum size=0.7em] {};
		\draw (8,3) node () [draw, rectangle, fill=white, minimum size=0.7em] {};
		\draw (8,5) node () [draw, rectangle, fill=white, minimum size=0.7em] {};
		\draw (8,9) node () [draw, rectangle, fill=white, minimum size=0.7em] {};
		
		\draw (1,4) node () [draw, rectangle, fill=white, minimum size=0.7em] {};
		\draw (5,4) node () [draw, rectangle, fill=white, minimum size=0.7em] {};
		\draw (6,4) node () [draw, rectangle, fill=white, minimum size=0.7em] {};
		\draw (10,4) node () [draw, rectangle, fill=white, minimum size=0.7em] {};
		\draw (11,4) node () [draw, rectangle, fill=white, minimum size=0.7em] {};
		
		\draw (1,6) node () [draw, rectangle, fill=white, minimum size=0.7em] {};
		\draw (2,6) node () [draw, rectangle, fill=white, minimum size=0.7em] {};
		\draw (5,6) node () [draw, rectangle, fill=white, minimum size=0.7em] {};
		\draw (7,6) node () [draw, rectangle, fill=white, minimum size=0.7em] {};
		\draw (9,6) node () [draw, rectangle, fill=white, minimum size=0.7em] {};
		
		\draw (3,4) node () [draw, circle, fill=white, minimum size=1.2em] {};
		\draw (3,6) node () [draw, circle, fill=white, minimum size=1.2em] {};
		\draw (4,4) node () [draw, circle, fill=white, minimum size=1.2em] {};
		\draw (4,6) node () [draw, circle, fill=white, minimum size=1.2em] {};
		\draw (8,4) node () [draw, circle, fill=white, minimum size=1.2em] {};
		\draw (8,6) node () [draw, circle, fill=white, minimum size=1.2em] {};
		
	\end{tikzpicture}
\end{tabular}
\end{center}
\caption{Patter of positions with non-zero entries for matrices having $ (X,Y) $ as a minimal cover. Here, $ s = |X| $ is the number of horizontal lines and $ r-s = |Y| $ is the number of vertical lines. The number of circles is thus $ s(r-s) $.}
\label{fig ball size upper bound}
\end{figure}

In this figure, horizontal lines correspond to the rows indexed by $ X $ and vertical lines correspond to columns indexed by $ Y $. A matrix in $ \mathbb{F}_q^{m \times n} $ having $ (X,Y) $ as a minimal cover would have zeros everywhere except in the lines shown in Fig. \ref{fig ball size upper bound}. Furthermore, in the positions depicted with circles, it may have any value from $ \mathbb{F}_q $. Since there are $ s(r-s) $ such circles, there are $ q^{s(r-s)} $ possibilities for such entries. On the other hand, when removing the positions depicted by circles from a given line, the remaining entries must form a nonzero vector (whose nonzero entries are depicted with squares). Since there are $ s $ horizontal lines and $ r-s $ vertical lines, this means that there are $ (q^{n-r+s}-1)^s $ possibilities for the horizontal lines and $ (q^{m-s}-1)^{r-s} $ possibilities for the vertical lines. In total, we have exactly $ (q^{m-s}-1)^{r-s}(q^{n-r+s}-1)^s q^{s(r-s)} $ matrices in $ \mathbb{F}_q^{m \times n} $ with $ (X,Y) $ as a minimal cover.

If we add such numbers, running over all cover patterns of cover weight $ r $, we obtain an upper bound on $ S_r^{m,n} $,
$$ {\rm UB}_r = \sum_{s=0}^r \binom{m}{s} \binom{n}{r-s} (q^{m-s} - 1)^{r-s} (q^{n-r+s}-1)^s q^{s(r-s)}. $$
In order to find the exact value of $ S_r^{m,n} $, we need to subtract from $ {\rm UB}_r $ all the matrices that we have double-counted when considering two of their minimal covers. This double-counting excess is the number $ {\rm DC}_r = {\rm UB}_r - S^{m,n}_r $.

Let $ (X,Y) , (X^\prime, Y^\prime) \in {\rm MC}(m,n) $ be such that $ (X,Y) \neq (X^\prime,Y^\prime) $ and $ |X| + |Y| = |X^\prime| + |Y^\prime| = r $. Let $ \omega = |X \cap X^\prime | + |Y \cap Y^\prime | $ be the number of lines that $ (X,Y) $ and $ ( X^\prime,Y^\prime) $ have in common. Clearly, $ 0 \leq \omega \leq r-1 $ since $ (X,Y) \neq (X^\prime,Y^\prime) $. Let $ u = | X\cap X^\prime | $ be the number of common rows, thus $ \omega - u = |Y \cap Y^\prime | $ is the number of common columns. In Fig. \ref{fig ball size double counting}, we represent the two covers $ (X,Y) $ and $ (X^\prime,Y^\prime) $.

\begin{figure} [!h]
\begin{center}
\begin{tabular}{c@{\extracolsep{1cm}}c}
	\begin{tikzpicture}[line width=1pt, scale=0.6]
		\tikzstyle{every node}=[inner sep=0pt, minimum width=4.5pt]
		

		\draw[thick] (2,0) -- (2,10);
		\draw[thick] (3,0) -- (3,10);
		\draw[thick, dashed] (6,0) -- (6,10);
		\draw[thick, dotted] (7,0) -- (7,10);
		\draw[thick, dashed] (9,0) -- (9,10);
		\draw[thick, dotted] (10,0) -- (10,10);
		
		\draw[thick, dashed] (0,9) -- (12,9);
		\draw[thick] (0,8) -- (12,8);
		\draw[thick] (0,7) -- (12,7);
		\draw[thick, dotted] (0,5) -- (12,5);
		\draw[thick, dashed] (0,3) -- (12,3);
		\draw[thick, dotted] (0,2) -- (12,2);
		
		\draw (2,8) node () [draw, cross out, fill=white, minimum size=0.7em] {};
		\draw (3,8) node () [draw, cross out, fill=white, minimum size=0.7em] {};
		\draw (2,7) node () [draw, cross out, fill=white, minimum size=0.7em] {};
		\draw (3,7) node () [draw, cross out, fill=white, minimum size=0.7em] {};
		
		\draw (2,9) node () [draw, circle, fill=white, minimum size=0.7em] {};
		\draw (3,9) node () [draw, circle, fill=white, minimum size=0.7em] {};
		\draw (2,5) node () [draw, circle, fill=white, minimum size=0.7em] {};
		\draw (3,5) node () [draw, circle, fill=white, minimum size=0.7em] {};
		\draw (2,3) node () [draw, circle, fill=white, minimum size=0.7em] {};
		\draw (3,3) node () [draw, circle, fill=white, minimum size=0.7em] {};
		\draw (2,2) node () [draw, circle, fill=white, minimum size=0.7em] {};
		\draw (3,2) node () [draw, circle, fill=white, minimum size=0.7em] {};
		
		\draw (6,8) node () [draw, circle, fill=white, minimum size=0.7em] {};
		\draw (6,7) node () [draw, circle, fill=white, minimum size=0.7em] {};
		\draw (7,8) node () [draw, circle, fill=white, minimum size=0.7em] {};
		\draw (7,7) node () [draw, circle, fill=white, minimum size=0.7em] {};
		\draw (9,8) node () [draw, circle, fill=white, minimum size=0.7em] {};
		\draw (9,7) node () [draw, circle, fill=white, minimum size=0.7em] {};
		\draw (10,8) node () [draw, circle, fill=white, minimum size=0.7em] {};
		\draw (10,7) node () [draw, circle, fill=white, minimum size=0.7em] {};
		
		\draw (6,2) node () [draw, diamond, fill=white, minimum size=0.9em] {};
		\draw (9,2) node () [draw, diamond, fill=white, minimum size=0.9em] {};
		\draw (6,5) node () [draw, diamond, fill=white, minimum size=0.9em] {};
		\draw (9,5) node () [draw, diamond, fill=white, minimum size=0.9em] {};
		
		\draw (7,3) node () [draw, rectangle, fill=white, minimum size=0.7em] {};
		\draw (10,3) node () [draw, rectangle, fill=white, minimum size=0.7em] {};
		\draw (7,9) node () [draw, rectangle, fill=white, minimum size=0.7em] {};
		\draw (10,9) node () [draw, rectangle, fill=white, minimum size=0.7em] {};
		
	\end{tikzpicture}
\end{tabular}
\end{center}
\caption{Patter of positions with non-zero entries for matrices having both $ (X,Y) $ and $ (X^\prime,Y^\prime) $ as minimal covers. Here, $ u = |X \cap X^\prime| $ and $ \omega - u = |Y \cap Y^\prime| $ are the numbers of solid horizontal and vertical lines, respectively, $ s = |X| - |X\cap X^\prime| $ and $ r-\omega-s = |Y| - |Y\cap Y^\prime| $ are the numbers of dashed horizontal and vertical lines, respectively, and $ t = |X^\prime| - |X\cap X^\prime| $ and $ r-\omega-t = |Y^\prime| - |Y\cap Y^\prime| $ are the numbers of dotted horizontal and vertical lines, respectively.}
\label{fig ball size double counting}
\end{figure}
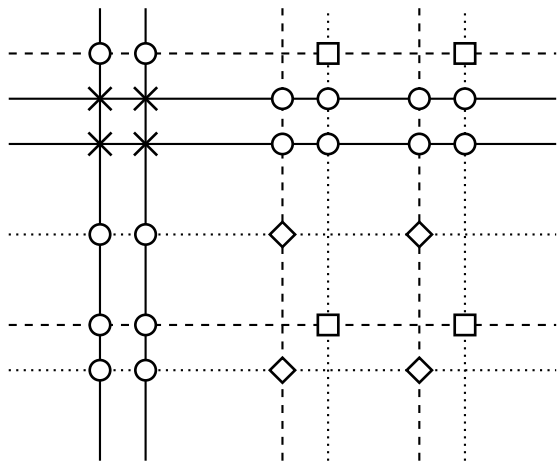

In this figure, horizontal lines, solid and dashed, correspond to the rows indexed by $ X $; horizontal lines, solid and dotted, correspond to rows indexed by $ X^\prime $; vertical lines, solid and dashed, correspond to columns indexed by $ Y $; and vertical lines, solid and dotted, correspond to columns indexed by $ Y^\prime $. In other words, solid lines belong to both $ (X,Y) $ and $ (X^\prime, Y^\prime) $.

Consider matrices in $ \mathbb{F}_q^{m \times n} $ having both $ (X,Y) $ and $ (X^\prime, Y^\prime) $ as minimal covers. First, such matrices would have zeros everywhere except in the solid lines and in the intersections between a dashed line and a dotted line, depicted with diamonds and squares in Fig. \ref{fig ball size double counting}. We now check how many possibilities there are for such entries for such matrices, thus counting the number of such matrices.

Such matrices may have any values in $ \mathbb{F}_q $ in the positions depicted with crosses (they do not determine whether a solid line is redundant or not). Since there are $ u(\omega - u) $ such crosses, there are $ q^{u(\omega - u)} $ possibilities for such entries. 

Next, the $ t(r-\omega-s) $ entries depicted with diamonds must form a matrix in $ \mathbb{F}_q^{t \times (r-\omega - s)} $ of maximum cover weight: If one row or column in such a $ t \times (r-\omega-s) $ matrix is zero, then the corresponding (dashed or dotted) line would be redundant. Thus there are $ S^{t, r-\omega-s}_{\min \{ t, r-\omega-s \}} $ possibilities for such entries. Analogously, there are $ S^{s, r-\omega-t}_{\min \{ s, r-\omega-t \}} $ possibilities for the $ s (r-\omega-t) $ entries depicted with squares. 

Finally, we need to check which values in the entries not depicted with crosses make a solid line non-redundant for both $ (X,Y) $ and $ (X^\prime,Y^\prime) $. Consider a vertical line. If the matrix contains nonzero entries in the indices not depicted with crosses or circles in Fig. \ref{fig ball size upper bound}, then the line is non-redundant for both $ (X,Y) $ and $ (X^\prime,Y^\prime) $. There are $ q^{t+s} (q^{m-u-s-t}-1) $ possibilities for such a case (the entries in the $ t+s $ depicted with circles may have any values, and the entries depicted with crosses are not considered in this paragraph). The other option is that all the entries in the indices not depicted with crosses or circles in Fig. \ref{fig ball size upper bound} are zero. In that case, the entries indexed by circles and part of a dotted line in Fig. \ref{fig ball size upper bound} must form a nonzero vector, since otherwise the solid line is redundant for the cover $ (X,Y) $. Similarly, the entries indexed by circles and part of a dashed line in Fig. \ref{fig ball size upper bound} must form a nonzero vector. There are $ (q^s-1)(q^t-1) $ possibilities for such a case. Thus, there are
$$ q^{t+s} (q^{m-u-s-t}-1) + (q^s-1)(q^t-1) $$
possibilities for a given vertical line to not be redundant for both $ (X,Y) $ and $ (X^\prime, Y^\prime) $. Similarly, there are 
$$ q^{(r-\omega-t)+(r-\omega-s)}(q^{n-(\omega-u)-(r-\omega-t)-(r-\omega-s)} - 1) + (q^{r-\omega-s}-1)(q^{r-\omega-t} - 1) $$
possibilities for a given horizontal line to not be redundant for both $ (X,Y) $ and $ (X^\prime,Y^\prime) $. Since the redundancy of a line (horizontal or vertical) is independent of any other line (horizontal or vertical), and there are $ \omega - u $ vertical lines and $ u $ horizontal lines, then there are
\begin{equation*}
\begin{split}
& \left( q^{t+s}(q^{m-u-s-t} - 1) + (q^s-1) (q^t-1) \right)^{\omega - u} \\
& \cdot \left( q^{(r-\omega-t)+(r-\omega-s)}(q^{n-(\omega-u)-(r-\omega-t)-(r-\omega-s)} - 1) + (q^{r-\omega-s}-1)(q^{r-\omega-t} - 1) \right)^u 
\end{split}
\end{equation*}
possibilities for all solid lines not to be redundant.

We have exactly counted the matrices in $ \mathbb{F}_q^{m \times n} $ that have both $ (X,Y) $ and $ (X^\prime,Y^\prime) $ as minimal covers. However, some of them may have other covers as minimal covers. Hence, the number of matrices that we have double-counted in $ {\rm UB}_r $, which is exactly $ {\rm DC}_r = {\rm UB}_r - S^{m,n}_r \geq 0 $, satisfies the upper bound in the theorem. Consider now the matrices that have both $ (X,Y) $ and $ (X^\prime, Y^\prime) $ as minimal covers, but all of whose entries depicted with diamonds and squares are nonzero. Clearly, those matrices cannot have a third cover of cover weight $ r $ as a minimal cover. Since these are only some of the double-counted matrices, our double-counting excess $ {\rm DC}_r $ satisfies the lower bound in the theorem.
\end{proof}

\end{document}